\documentclass[a4paper,UKenglish,cleveref,thm-restate]{lipics-v2021}
\usepackage{float}
\usepackage{refcount}

\nolinenumbers

\newcommand{\Q}{\mathbb Q}
\newcommand{\R}{\mathbb R}
\newcommand{\Z}{\mathbb Z}

\renewcommand{\O}{\mathcal O}
\newcommand{\dbb}[2]{\mathcal B_*^{#1\times#1}(#2)}
\newcommand{\db}[2]{\mathcal B^#1(#2)}
\DeclareMathOperator{\cone}{cone}
\DeclareMathOperator{\intcone}{int.cone}
\DeclareMathOperator{\lcm}{lcm}
\newcommand{\rel}{{\bowtie}}
\newcommand{\graverbasis}{\mathcal G}
\newcommand{\signs}{\{<,=,>\}}
\newcommand{\soverline}[1]{\smash{\overline{#1}}}

\newcounter{theorem1}
\theoremstyle{claimstyle}
\newtheorem{claimin}{Claim}[theorem]
\newenvironment{proofof*}[1]{\setcounter{theorem1}{\value{theorem}}\setcounter{theorem}{\getrefnumber{#1}}\begin{proof}}{\end{proof}\setcounter{theorem}{\value{theorem1}}}
\newenvironment{proofof}[1]{\setcounter{theorem1}{\value{theorem}}\setcounter{theorem}{\getrefnumber{#1}}\begin{proof}[Proof of \cref{#1}]}{\end{proof}\setcounter{theorem}{\value{theorem1}}}

\makeatletter
\newenvironment{cdisplaymath}{\@fleqnfalse\begin{displaymath}}{\end{displaymath}}
\makeatother

\crefname{lemma}{Lemma}{Lemmata}
\Crefname{lemma}{Lemma}{Lemmata}
\crefname{ilp}{ILP}{ILPs}
\Crefname{ilp}{ILP}{ILPs}
\creflabelformat{ilp}{(#2#1#3)}
\crefname{constraint}{Constraint}{Constraints}
\Crefname{constraint}{Constraint}{Constraints}
\creflabelformat{constraint}{(#2#1#3)}

\title{Solving 4-Block Integer Linear Programs Faster Using Affine Decompositions of the Right-Hand Sides}

\titlerunning{Solving 4-Block ILPs Faster Using Affine Decompositions of the Right-Hand Sides}

\author{Alexandra Lassota}{Eindhoven University of Technology, Netherlands}{a.a.lassota@tue.nl}{https://orcid.org/0000-0001-6215-066X}{}

\author{Koen Ligthart}{Eindhoven University of Technology, Netherlands}{k.m.ligthart@tue.nl}{https://orcid.org/0009-0004-6823-5225}{}

\authorrunning{A. Lassota and K. Ligthart}

\Copyright{Alexandra Lassota and Koen Ligthart} 



\hideLIPIcs

\begin{document}

\maketitle

\begin{abstract}
We present a new and faster algorithm for the 4-block integer linear programming problem, overcoming the long-standing runtime barrier faced by previous algorithms that rely on Graver complexity or proximity bounds. The 4-block integer linear programming problem asks to compute 
$\min\bigl\{c_0^\top x_0+c_1^\top x_1+\dots+c_n^\top x_n\bigm\vert Ax_0+Bx_1+\dots+Bx_n=b_0,\ Cx_0+Dx_i=b_i\ \forall i\in[n],\ (x_0,x_1,\dots,x_n)\in\Z_{\ge0}^{(1+n)k}\bigr\}$ for some $k\times k$ matrices $A,B,C,D$ with coefficients bounded by $\overline\Delta$ in absolute value. Our algorithm runs in time $f(k,\overline\Delta)\cdot n^{k+\O(1)}$, improving upon the previous best running time of $f(k,\overline\Delta)\cdot n^{k^2+\O(1)}$~[Oertel, Paat, and Weismantel (Math. Prog. 2024), Chen, Koutecký, Xu, and Shi (ESA 2020)]. Further, we give the first algorithm that can handle large coefficients in $A, B$ and $C$, that is, it has a running time that depends only polynomially on the encoding length of these coefficients. We obtain these results by extending the $n$-fold integer linear programming algorithm of Cslovjecsek, Koutecký, Lassota, Pilipczuk, and Polak (SODA 2024) to incorporate additional global variables $x_0$. The central technical result is showing that the exhaustive use of the vector rearrangement lemma of Cslovjecsek, Eisenbrand, Pilipczuk, Venzin, and Weismantel (ESA 2021) can be made \emph{affine} by carefully guessing both the residue of the global variables modulo a large modulus and a face in a suitable hyperplane arrangement among a sufficiently small number of candidates. This facilitates a dynamic high-multiplicy encoding of a \emph{faithfully decomposed} $n$-fold ILP with bounded right-hand sides, which we can solve efficiently for each such guess.
\end{abstract}

\section{Introduction}

We study the complexity of solving integer linear programs (ILPs) for a prominent class of ILPs that have a special block-structure in the constraint matrix, called \emph{4-block} ILPs. The general integer linear programming problem (in standard form) asks to compute an optimal solution $x$ to
\[
    \min\bigl\{c^\top x\bigm\vert Qx=b,x\in\Z_{\ge0}^n\bigr\},
\]
for some constraint matrix $Q\in\Z^{m\times n}$, right-hand side vector $b\in\Z^m$, and objective vector $c\in\Z^n$. This problem is known to be NP-hard for general $Q$, which has catalyzed broad research to understand the (parameterized) complexity of solving ILPs with restricted classes of constraint matrices; see, e.g.,~\cite{DBLP:journals/disopt/GavenciakKK22} for a survey. Among those, we focus on \emph{4-block} matrices, which are of the form
\[
    Q=\begin{pmatrix}
        A&B_1&B_2&\dots&B_n\\
        C_1&D_1&&&\\
        C_2&&D_2&&\\
        \vdots&&&\ddots&\\
        C_n&&&&D_n
    \end{pmatrix}
\]
for blocks $A\in\Z^{m\times s},B_i\in\Z^{m\times t},C_i\in\Z^{d\times s},D_i\in\Z^{d\times t}$ for $i\in[n]:=\{1,2,\dots,n\}$. Here, blocks of zeroes are omitted. The underlying 4-block ILP problem asks to solve an ILP which has a constraint matrix that is a 4-block matrix.

By splitting up $x,c$ and $b$ into $1+n$ bricks matching the block-structure of a 4-block matrix $Q$, solving the 4-block ILP problem corresponds to finding an optimal solution to
\begin{equation}
    \begin{aligned}
        \min\bigl\{c_0^\top x_0+c_1^\top x_1+\dots+c_n^\top x_n\bigm\vert&\ Ax_0+B_1x_1+\dots+B_nx_n=b_0,\\
        &\ C_ix_0+D_ix_i=b_i\ \forall i\in[n],\\
        &\ (x_0,x_1,\dots,x_n)\in\Z_{\ge0}^{s+nt}\bigr\}.
    \end{aligned}
    \label[ilp]{ilp:4-block}\tag{4B}
\end{equation}
We use $\overline\Delta$ to denote the greatest absolute value of the coefficients of $Q$ and use $\Delta\le\overline\Delta$ to denote the greatest absolute value of the coefficients in the diagonal blocks $D_i$. Furthermore, we let $L$ denote the encoding length of the largest absolute value of a number in \cref{ilp:4-block}. We say that a 4-block problem is uniform if $B_1=B_2=\dots=B_n,C_1=C_2=\dots=C_n$ and $D_1=D_2=\dots=D_n$ and say that a 4-block problem is $B$-uniform if $B_1=B_2=\dots=B_n$.
In the following, $\O_x(y)$ is used to denote $\O(f(x)\cdot y)$ for some (computable) function $f$, i.e., this hides multiplicative factors depending on the parameter(s) $x$. 

It is known that 4-block ILPs can be solved in polynomial time when the dimensions and coefficients of all blocks are bounded by a constant~\cite{DBLP:conf/esa/0009K0S20,DBLP:conf/ipco/HemmeckeKW10,DBLP:journals/mp/OertelPW24}, i.e., the problem is in XP when parameterized by $d,m,s,t,\overline\Delta$. Hemmecke, Köppe, and Weismantel~\cite{DBLP:conf/ipco/HemmeckeKW10} were the first to show this by providing an algorithm that solves a uniform 4-block ILP in slice-wise polynomial time $\O_{dmst\soverline\Delta}(n^{s2^m+\O(1)}\cdot L^{\O(1)})$. Subsequent works have improved the dependence of the block dimensions on the exponent of $n$ in the running time: Chen et al.~\cite{DBLP:conf/esa/0009K0S20} solve uniform 4-block ILPs in time $\O_{dmst\soverline\Delta}(n^{sm+\O(1)})$ and Oertel, Paat, and Weismantel~\cite{DBLP:journals/mp/OertelPW24} solve (non-uniform) 4-block ILPs in time\footnote{This running time can be obtained by using the algorithm in~\cite{DBLP:conf/soda/CslovjecsekEHRW21} and a strongly polynomial algorithm, such as Tardos'~\cite{DBLP:journals/ior/Tardos86} algorithm, to solve the LP relaxation under the assumption that $\overline\Delta$ is small.} $\O_{dmst\soverline\Delta}(n^{s\cdot\min\{s+2,m\}+\O(1)})$. Furthermore, attention has been given to the following special cases: Chen et al.~\cite{DBLP:conf/esa/0009K0S20} provide an algorithm that solves uniform 4-block ILPs in time $\O_{dmst\soverline\Delta}(n^{s(t^2+1)+\O(1)})$ when $A$ is the zero matrix and Chen, Chen, and Zhang~\cite{DBLP:journals/mp/ChenCZ24} provide an algorithm solving 4-block ILPs in time $\O_{dmst\soverline\Delta}(n^{\O(1)})$ as long as all $C_i$ are uniformly equal to some rank 1 matrix.

All previously mentioned algorithms share the common methodology of proving a bound~$u$ on either the proximity, which limits the distance between optimal LP relaxation solutions and their nearest optimal integer solutions, or the Graver complexity, which represents the granularity of the lattice of integer kernel elements of the constraint matrix. These bounds are used to reduce the 4-block ILP problem to solving a polynomial number of 4-block ILPs in which the domain of $x_0$ is restricted to a subset of the box $[0,2u]^s$. This makes it possible to eliminate the global variable vector $x_0$ by guessing its value, which is one of the $\O(u)^s$ integer vectors in the box\footnote{We note that under the strong assumption that all $C_i$ are equal to some rank 1 matrix $C$, the authors of~\cite{DBLP:journals/mp/ChenCZ24} instead enumerate the options for $Cx$ in a 1-dimensional box, similar to the situation where $s=1$.}. Finally, the resulting 4-block ILP problems without global variables ($s=0$), typically referred to as $n$-fold ILPs, can be solved efficiently in $\O_{dmst\soverline\Delta}(n^{\O(1)})$ time~\cite{DBLP:conf/soda/CslovjecsekEHRW21}.

Since Chen et al.~\cite{DBLP:conf/esa/0009K0S20} show that the Graver complexity (and, as an analogous construction can be used, also the proximity) of 4-block constraint matrices may be as large as $u=\Omega_{dmst\soverline\Delta}(n^{\min\{d,m,s,t\}})$, such box guessing strategies appear to be limited to a quadratic dependence of $\Omega(s\cdot\min\{d,m,s,t\})$ of the block dimensions on the exponent of $n$ in the running time. On the contrary, Eisenbrand and Rothvoss~\cite{eisenbrand2025parameterizedlinearformulationinteger} conjecture that the 4-block ILP problem is fixed-parameter tractable (FPT).

\begin{conjecture}[\cite{eisenbrand2025parameterizedlinearformulationinteger}]
    Solving 4-block ILP is FPT when parameterized by $d,m,s,t,\overline\Delta$, i.e., it can be solved in time $\O_{dmst\soverline\Delta}((nL)^{\O(1)})$.
    \label{conjecture:4-block}
\end{conjecture}

The lower bound by Chen et al.~\cite{DBLP:conf/esa/0009K0S20} shows that in order to answer \cref{conjecture:4-block} in the affirmative, different algorithmic techniques are required. In the context of \cref{conjecture:4-block}, Cslovjecksek et al.~\cite{DBLP:journals/theoretics/CslovjecsekKLPP25} show that uniform 4-block ILP is FPT w.r.t.\ $d,m,s,t,\overline\Delta$ if and only if uniform 4-block ILP is FPT w.r.t.\ $d,m,s,t,\Delta$ and $\overline\Delta\le n$: a proof of \cref{conjecture:4-block} would imply that solving uniform 4-block ILPs with moderately large coefficients in the non-diagonal blocks is also FPT. For this reason, algorithms that are parameterized by $\Delta$ instead of $\overline\Delta$ and of which the running time grows only polynomially with $L$ are of interest\footnote{Parameterization by the coefficient size $\Delta$ of the diagonal blocks $D_i$ is known to be necessary~\cite{DBLP:journals/disopt/ChenCZ22,DBLP:journals/ai/DvorakEGKO21}.}. Such results seem out of reach for a box-guessing approach as the Graver complexity and proximity of ILPs are typically at least as large as the coefficients of the constraint matrix.

Cslovjecksek et al.~\cite{DBLP:journals/theoretics/CslovjecsekKLPP25}, followed by Eisenbrand and Rothvoss~\cite{eisenbrand2025parameterizedlinearformulationinteger}, provide FPT algorithms able to deal with large coefficients outside the diagonal blocks for two intensively studied special cases of 4-block ILPs: \emph{$n$-fold} and \emph{2-stage-stochastic} ILPs, which we discuss next. As shown in \cref{fig:block-structured}, these have no global variables, i.e., $s=0$, or no global constraints, i.e., $m=0$, respectively. Yet, their expressiveness is sufficient to be applied in the context of graph problems~\cite{DBLP:conf/isaac/FellowsLMRS08,DBLP:journals/dam/FialaGKKK18}, scheduling~\cite{DBLP:journals/mp/JansenKMR22,DBLP:journals/scheduling/KnopK18}, computational social choice~\cite{bartholdi1989voting,DBLP:journals/teco/KnopKM20}, and machine learning~\cite{ermolieva2023connections}. See also~\cite{DBLP:journals/disopt/GavenciakKK22}.

\begin{figure}[H]
    \centering
    \vspace{-5pt}
    \begin{tabular}{p{0.5\linewidth}@{\hskip0pt}p{0.5\linewidth}}
        \normalsize
        \begin{cdisplaymath}
            \begin{pmatrix}
                B_1&B_2&\dots&B_n\\
                D_1&&&\\
                &D_2&&\\
                &&\ddots&\\
                &&&D_n
            \end{pmatrix}
        \end{cdisplaymath}
        \vspace{-15pt} &
        \normalsize
        \begin{cdisplaymath}
            \begin{pmatrix}
                C_1&D_1&&&\\
                C_2&&D_2&&\\
                \vdots&&&\ddots&\\
                C_n&&&&D_n
            \end{pmatrix}
        \end{cdisplaymath}
        \vspace{-15pt}\\
        \subcaption{$n$-fold structure ($s=0$).} &
        \subcaption{2-stage stochastic structure ($m=0$).}
    \end{tabular}
    \vspace{-15pt}
    \caption{Displayed are two prominent special cases of 4-block matrices.}
    \label{fig:block-structured}
\end{figure}

Current state-of-the-art FPT algorithms for these special cases in terms of their polynomial dependence of $n$ solve $n$-fold ILPs in time $\O_{dm\soverline\Delta}((nt)^{1+o(1)})$~\cite{DBLP:conf/soda/CslovjecsekEHRW21} and 2-stage stochastic ILPs in time $\O_{st\soverline\Delta}((nd)^{1+o(1)})$~\cite{DBLP:conf/esa/CslovjecsekEPVW21}. The dependence on $\overline\Delta$ in the running times of these algorithms is a result of the use of proximity bounds.

Cslovjecksek et al.~\cite{DBLP:journals/theoretics/CslovjecsekKLPP25} avoid the use of such bounding arguments and solve $B$-uniform\footnote{Note that due to the NP-hardness of the subset sum problem, the $B$-uniformity in the setting of large coefficients is necessary.} $n$-fold ILP in time $\O_{dt\Delta}(m\cdot(nL)^{\O(1)})$ and 2-stage stochastic ILP feasibility in time $\O_{ds\Delta}(t\cdot(nL)^{\O(1)})$. Eisenbrand and Rothvoss~\cite{eisenbrand2025parameterizedlinearformulationinteger} extend the former result by providing an algorithm that optimizes over a 2-stage stochastic ILP in time $\O_{dst\Delta}((nL)^{\O(1)})$ and, using similar arguments, show how to find an \emph{almost solution} to a 4-block ILP, that violates the $m$ global constraints by an additive error of at most $\O_{dt\soverline\Delta}(m)$, in time $\O_{dst\Delta}(m\cdot(nL)^{\O(1)})$.

\subsection{Contributions}

We extend the $B$-uniform $n$-fold ILP algorithm from~\cite{DBLP:journals/theoretics/CslovjecsekKLPP25} to obtain an algorithm that solves $B$-uniform 4-block ILPs with an improved running time dependence of $s+\O(1)$ in the exponent of $n$, improving over previous quadratic dependencies. Furthermore, this algorithm is the first whose running time is polynomial in the encoding lengths of the coefficients of $A,B$ and $C_i$. With this, we overcome both limitations inherent to the enumeration of vectors in an $s$-dimensional box with side length given by the Graver complexity or proximity.

\begin{restatable}{theorem}{theoremuniformfourblockxp}
    A $B$-uniform 4-block ILP can be solved in time $\O_{dst\Delta}(n^{s+\O(1)}\cdot m\cdot L^{\O(1)})$.
    \label{theorem:uniform-4-block-xp}
\end{restatable}

The constant hidden in the $\O_{dst\Delta}$ is triply exponential in $d$, which matches that of the $n$-fold algorithm by Cslovjecksek et al.~\cite{DBLP:journals/theoretics/CslovjecsekKLPP25}.

The algorithm underlying \cref{theorem:uniform-4-block-xp} relies on constructing \emph{faithful decompositions} of right-hand side vectors $b_i$ that play a central role in the $n$-fold ILP algorithm of Cslovjecksek et al.~\cite{DBLP:journals/theoretics/CslovjecsekKLPP25}. In this setting, a faithful decomposition may be interpreted as a decomposition of the local bricks of an $n$-fold ILP into more bricks that have only a bounded number of distinct and small right-hand side vectors. We extend their approach by computing faithful decompositions in a dynamic and affine manner for varying $b_i$. This is made precise in \cref{sec:preliminaries,sec:dynamic-faithful-decomposition}. To provide such decompositions, we exploit the constructive nature of the proof of Cslovjecksek et al.~\cite{DBLP:conf/esa/CslovjecsekEPVW21} of their variant of a vector rearrangement lemma by Klein~\cite{DBLP:journals/mp/Klein22}. The affine faithful decompositions are used to encode a domain-restricted 4-block ILP for a limited number of guesses of the domain of the global variables. The reduced dependence on the exponent of $n$ in \cref{theorem:uniform-4-block-xp} relies on the fact that the number of faces in an arrangement of $N$ hyperplanes in $\R^s$ is bounded by $\O(N)^s$~\cite{DBLP:series/eatcs/Edelsbrunner87,grunbaum1967convex} and can be computed in $\O(N)^s$ time~\cite{DBLP:series/eatcs/Edelsbrunner87,DBLP:journals/siamcomp/EdelsbrunnerOS86,DBLP:journals/siamcomp/EdelsbrunnerSS93}.

We also obtain an improved algorithm for (non-$B$-uniform) 4-block ILPs with bounded coefficients through a reduction to the uniform case by Dvořák et al.~\cite{DBLP:journals/ai/DvorakEGKO21}. By additionally applying Frank and Tardos'~\cite{DBLP:journals/combinatorica/FrankT87} coefficient reduction technique, we obtain the slice-wise strongly polynomial-time algorithm in \cref{theorem:4-block-strongly-xp}.

\begin{restatable}{theorem}{theoremfourblockstronglyxp}
    A 4-block ILP can be solved in time $\O_{dms\soverline\Delta}(n^{s+\O(1)}\cdot t)$.
    \label{theorem:4-block-strongly-xp}    
\end{restatable}

We present the preliminaries in \cref{sec:preliminaries}. In \cref{sec:algorithm}, we discuss how affine faithful decompositions can be used to obtain an algorithm to solve 4-block ILPs. We prove the lemma that constructs the used decompositions in \cref{sec:dynamic-faithful-decomposition}. In the latter two respective sections, we give a high level overview of the used techniques before providing detailed proofs.

\section{Preliminaries}
\label{sec:preliminaries}

When $A$ is a matrix or scalar and $S\subseteq\R^d$ is a set of vectors, we use $AS$ to denote the image $\{A\cdot x\vert x\in S\}$ of $S$ under $A$. Similarly, we use $S+v$ to denote the translated set~$\{x+v\vert x\in S\}$ for an offset vector $v\in\R^d$. For notational convenience, we identify a set of vectors $S\subseteq\R^d$ with the $d\times|S|$ matrix of which the columns are the vectors in $S$. The \emph{cone} generated by a set of vectors $S=\{v_1,\dots,v_n\}$ is the set of nonnegative combinations of vectors in $S$ and the \emph{integer cone} generated by these vectors is the set of nonnegative integral combinations. That is,
\[
    \cone S=\{\lambda_1\cdot v_1+\dots+\lambda_n\cdot v_n\ \vert\ \lambda\in\R_{\ge0}^n\}=S\R_{\ge0}^S\text{ and }\intcone S=S\Z_{\ge0}^S.
\]
Here, we use $K^I$ for sets $K$ and $I$ to denote the set of vectors $x\in K^I$ that are indexed by $i\in I$ and have coordinates $x_i\in K$. As shorthand, we use $\db d\Delta=\{x\in\Z^d:\|x\|_\infty\le\Delta\}$ to denote the discrete radius $\Delta$ hypercube centered at the origin. Further, $\dbb d\Delta=\{F\in\{-\Delta,-\Delta+1,\dots,\Delta\}^{d\times d}:\det F\ne0\}$ denotes the set of all invertible $d\times d$ matrices $F$ with integral entries of absolute value bounded by $\Delta$, or equivalently, the set of all bases $F\subseteq\db d\Delta$ of $\R^d$. We say that a vector $x\in\R^d$ is \emph{conformal} to a vector $y\in\R^d$ if $x_iy_i\ge0$ and $|x_i|\le|y_i|$ for all $i\in[d]$. We denote this by $x\sqsubseteq y$. 

We use a $(d+1)$-dimensional point $(a_h,\beta_h)\in\R^d\times\R$ to describe a \emph{hyperplane} $h$ that represents the set $\{x\in\R^d:a_h^\top x=\beta_h\}$. Note that we allow $a_h=0$, i.e., degenerate hyperplanes that are either the whole space or the empty set. A (partially open) \emph{polyhedron}~$P$ is defined by its relation to a finite number of hyperplanes $H_\le,H_<$ and can be written as $P=\{x\in\R^d:a_h^\top x\le\beta_h\ \forall h\in H_\le,a_h^\top x<\beta_h\ \forall h\in H_<\}$. Similarly, a collection of hyperplanes $H\subseteq\R^{d+1}$ induces an \emph{arrangement} of \emph{faces} $\mathcal F_H$ each defined by a \emph{position vector} $\rel$ that specifies the position of the face w.r.t.\ each hyperplane. Here, the (potentially empty) face $F_\rel$ with position vector $\rel\in\signs^H$ is the partially open polyhedron given by $F_\rel=\{x\in\R^d:a_h^\top x\mathrel{\rel_h}\beta_h\ \forall h\in H\}$, and the arrangement induced by $H$ consists of all faces $\mathcal F_H=\{F_\rel\ \vert\ \rel\in\signs^H\}$. It is known that the number of distinct faces in a hyperplane arrangement in $d$ dimensions is bounded by $\O(|H|)^d$~\cite{DBLP:series/eatcs/Edelsbrunner87,grunbaum1967convex}. Furthermore, Edelsbrunner, O'Rourke, and Seidel~\cite{DBLP:journals/siamcomp/EdelsbrunnerOS86} have shown that arrangements can be computed within a similar amount of time.

\begin{theorem}[\cite{DBLP:series/eatcs/Edelsbrunner87,DBLP:journals/siamcomp/EdelsbrunnerOS86,DBLP:journals/siamcomp/EdelsbrunnerSS93}]
    Let $H$ be a collection of hyperplanes in $d$ dimensions. The arrangement $\mathcal F_H$ induced by $H$ can be computed in $\O(|H|)^d$ time.
    \label{theorem:edelsbrunner}
\end{theorem}

In particular, the algorithm of Edelsbrunner et al.\ can provide the at most $\O(|H|)^d$ many position vectors that correspond to nonempty faces.

A central concept in the $n$-fold algorithm by Cslovjecksek et al.~\cite{DBLP:journals/theoretics/CslovjecsekKLPP25} and the algorithm underlying \cref{theorem:uniform-4-block-xp} is that of \emph{faithful decompositions} of right-hand side vectors $b$ into smaller vectors, which captures the fact that solutions to ILPs in standard form may be decomposed as sums of solutions to ILPs with bounded right-hand sides.

\begin{definition}
    A faithful decomposition of $b\in\Z^d$ w.r.t.\ $D\in\Z^{d\times t}$ of order $\Xi$ is a multiset of vectors $b_1,\dots,b_\ell\in\Z^d$ such that
    \begin{itemize}
        \item $b$ decomposes as $b=b_1+\dots+b_\ell$;
        \item each $b_i\sqsubseteq b$ for $i\in[\ell]$;
        \item each $b_i$ has norm bounded by $\|b_i\|_\infty\le\Xi$ for $i\in[\ell]$; and
        \item any solution $x\in\Z_{\ge0}^t$ to $Dx=b$, can be decomposed as $x=x_1+\dots+x_\ell$ where each $x_i\in\Z_{\ge0}^t$ is a solution to $Dx_i=b_i$ for $i\in[\ell]$;
    \end{itemize}
    \label{def:faithful-decomposition}
\end{definition}

A multiset consisting of elements in a set $S$ may be interpreted as an integer vector~$f\in\Z_{\ge0}^S$ where $f_s$ denotes the multiplicity with which $s\in S$ occurs in the multiset. For this reason, we identify multisets with integer vectors. Let $X_{b'}^D=\{x\in\Z_{\ge0}^t:Dx=b'\}$ be the set of solutions for a given right-hand side $b'$. Translating \cref{def:faithful-decomposition} reveals that a faithful decomposition of $b$ w.r.t.\ $D$ of order $\Xi$ is an integer vector $f\in\Z^R$ for some $R\subseteq\db d\Xi$ such that
\begin{itemize}
    \item $b=\sum_{b'\in R}f_{b'}\cdot b'$;
    \item $b'\sqsubseteq b$ for $b'\in R$;
    \item for all $x\in X_b^D$, there exists a multiset $p_{b'}\in\Z_{\ge0}^{X_{b'}^D}$ for each $b'\in R$ such that $x=\sum_{b'\in R}\sum_{x'\in X_{b'}^D}(p_{b'})_{x'}\cdot x'$ and $\sum_{x'\in X_{b'}^D}(p_{b'})_{x'}=f_{b'}$ for $b'\in R$.
\end{itemize}
This high-multiplicity encoding is necessary as low-order faithful decompositions need to have a multiset cardinality that is exponential in the encoding length of $b$. Cslovjecksek et al.~\cite{DBLP:journals/theoretics/CslovjecsekKLPP25} provide an algorithm (Lemma 5.12) that computes a (high-multiplicity encoding of a) faithful decomposition of $b\in\Z^d$ w.r.t.\ a $D\in\{-\Delta,-\Delta+1,\dots,\Delta\}^{d\times t}$ of order $2^{(d\Delta)^{\O(d)}}$ in time $\O_{dt\Delta}(\log^{\O(1)}\|b\|_\infty)$. We show in \cref{lemma:dynamic-decomposition} that faithful decompositions can be computed as an \emph{affine} function of $b$ as long as $b$ belongs to a well-behaved domain, which allows us to extend the algorithm by Cslovjecksek et al.~\cite{DBLP:journals/theoretics/CslovjecsekKLPP25} to solve 4-block ILPs. We note that the conformality of the decomposition in \cref{def:faithful-decomposition} is not algorithmically exploited in our algorithm, but is included in \cref{def:faithful-decomposition} to match the definitions of~\cite{DBLP:journals/theoretics/CslovjecsekKLPP25}.

To prove \cref{theorem:uniform-4-block-xp}, we will need \cref{lemma:solution-decomposition} from~\cite{DBLP:journals/theoretics/CslovjecsekKLPP25}, which shows that an arbitrarily large solution to an ILP with bounded coefficients can be decomposed into a small base solution plus small elements from the \emph{Graver basis} of the constraint matrix. The Graver basis~$\graverbasis(D)\subseteq\Z^t\setminus\{0\}$ of an integer matrix $D\in\Z^{d\times t}$ is the set of conformally minimal nonzero integral kernel elements of $D$. It is known that any Graver basis element $g\in\graverbasis(D)$ is bounded by $\|g\|_\infty\le\|g\|_1\le\O(d\Delta)^d$ when $D\in\{-\Delta,-\Delta+1,\dots,\Delta\}^{d\times t}$~\cite{DBLP:conf/icalp/EisenbrandHK18}.

\begin{lemma}[Lemma 5.3 in~\cite{DBLP:journals/theoretics/CslovjecsekKLPP25}, \cite{pottier1991minimal}]
    Let $D\in\{-\Delta,-\Delta+1,\dots,\Delta\}^{d\times t},b\in\db d\Xi$ and $x\in\Z_{\ge0}^t$ be so that $Dx=b$. Then there is a solution $\underline x\in\Z_{\ge0}^t$ to $D\underline x=b$ bounded by $\|\underline x\|_\infty\le\O(d(\Delta+\Xi))^d$ and a multiset $\mu\in\Z_{\ge0}^{\smash{\graverbasis(D)\cap\Z_{\ge0}^t}}$ of nonnegative Graver basis elements such that $x=\underline x+\sum_{g\in\graverbasis(D)\cap\Z_{\ge0}^t}\mu_g\cdot g$.
    \label{lemma:solution-decomposition}
\end{lemma}

\section{Algorithm}
\label{sec:algorithm}

This section presents the $\O_{dst\Delta}(n^{s+\O(1)}\cdot m\cdot L^{\O(1)})$ time algorithm for 4-block ILPs. Before describing the algorithm in detail, we sketch how Cslovjecksek et al.~\cite{DBLP:journals/theoretics/CslovjecsekKLPP25} employ faithful decompositions to solve (the feasibility variant of) the $B$-uniform $n$-fold ILP problem and argue how this can be extended to 4-block ILPs. In essence, Cslovjecksek et al.~\cite{DBLP:journals/theoretics/CslovjecsekKLPP25} preprocess a given $n$-fold ILP by splitting each brick $D_ix_i=b_i$ using a faithful-decomposition $b_i=b_i^{(1)}+\dots+b_i^{(\ell)}$ into the bricks $D_ix_i^{(1)}=b_i^{(1)},\dots,D_ix_i^{(\ell)}=b_i^{(\ell)}$. This results in an exponentially sized $n$-fold ILP which consists of only $\O_{dt\Delta}(1)$ different bricks with bounded right-hand sides. To exploit this latter fact, they use \cref{lemma:solution-decomposition} to derive a high-multiplicity encoding of the ILP in the form of a \emph{configuration} ILP that has only $\O_{dt\Delta}(1)$ integral variables and contains the high-multiplicity representations of the faithful decompositions as right-hand sides.

Now, observe that a 4-block ILP may be viewed as an ``$n$-fold ILP'' where the right-hand sides $b_i'$ depend affinely on the global variables $x_0$ through $b_i'=L_i(x_0):=b_i-C_ix_0$. We combine this observation with our central \cref{lemma:dynamic-decomposition}, which shows that if $b_i'$ belongs to a fixed lattice translate $r+M\Z^d$ and fixed partially open polyhedron $F$ among $\O_{d\Delta}(1)$ many options, there exist faithful decompositions of $b_i'$ that depend affinely on $b_i'$. The proof of \cref{lemma:dynamic-decomposition} is postponed to \cref{sec:dynamic-faithful-decomposition}.

\begin{restatable}{lemma}{lemmadynamicdecomposition}
    Let $d$ and $\Delta\in\Z_{\ge0}$ be given. Then, there exists a modulus $M=2^{\O(d\Delta)^{d^3}}$, a number $\Xi=2^{\O(d\Delta)^d}$, and an algorithm that outputs:
    \begin{itemize}
        \item A collection $H$ of at most $\O(d\Delta)^{d^4}$ hyperplanes;
        \item An affine map $\mathcal D_{r,F}:\R^d\to\R^{R_{r,F}}$ with $R_{r,F}\subseteq\db d\Xi$ and $|R_{r,F}|\le d+1$ for every remainder vector $r\in\{0,1,\dots,M-1\}^d$ and face $F\in\mathcal F_H$ in the arrangement induced by $H$.
    \end{itemize}
    For all $b\in(r+M\Z^d)\cap F$ and $D\in\{-\Delta,-\Delta+1,\dots,\Delta\}^{d\times t}$, the image $\mathcal D_{r,F}(b)$ is a faithful decomposition of $b$ w.r.t.\ $D$ of order $\Xi$. The algorithm outputs the hyperplanes in $\O(d\Delta)^{d^4}$ time and constructs an affine map $\mathcal D_{r,F}$ for a given $r$ and $F$ in $\O(\Delta)^{d^4}$ time.
    \label{lemma:dynamic-decomposition}
\end{restatable}

In our setting, $b_i'$ depends affinely on $x_0$, meaning that we can construct high-multiplicity encoded faithful decompositions of $b_i'$ that depend affinely on $x_0$. As this encoding appears in the right-hand side of the configuration ILP of Cslovjecksek et al.~\cite{DBLP:journals/theoretics/CslovjecsekKLPP25}, we can replace the static faithful decompositions obtained with the $\O_{dt\Delta}(\log^{\O(1)}\|b\|_\infty)$ time algorithm by Cslovjecksek et al.~\cite{DBLP:journals/theoretics/CslovjecsekKLPP25} with affine functions of $x_0$ and obtain a configuration integer program, that encodes the 4-block ILP, with linear constraints. This ILP, which has $\O_{dt\Delta}(1)+s$ integral variables, can then be solved using Kannan's algorithm~\cite{DBLP:journals/mor/Kannan87}.

The configuration ILP can be constructed as long as the right-hand sides $b_i'$ belong to a domain of the form $(r+M\Z^d)\cap F$ that is known upfront for each brick~$i\in[n]$, c.f.\ \cref{lemma:dynamic-decomposition}. To achieve this, we simultaneously guess the domains for each brick. A key observation is that the structure of the domains enables us to guess this with significantly less than the trivial bound of $\O_{d\Delta}(1)^n$ guesses by ``lifting'' these domains to the $s$-dimensional space of the global variables $x_0$.

The lattice structure $r+M\Z^d$ is especially convenient, as it can be guessed with a parameterized number of guesses, which has recently been exploited in $2$-stage stochastic integer linear programming~\cite{DBLP:journals/theoretics/CslovjecsekKLPP25,eisenbrand2025parameterizedlinearformulationinteger}: guessing the remainder vector $r\in\{0,1,\dots,M-1\}^s$ for which $x_0\equiv r\pmod M$ automatically fixes the remainder vector $r_i$ of $L_i(x_0)$ modulo $M$ to $r_i=(b_i-C_ir)\bmod M$. Here, we abuse notation to perform modular arithmetic element-wise on vectors.

Simultaneously guessing the faces $F$ is more expensive. For each hyperplane $h$ in the set generated by \cref{lemma:dynamic-decomposition}, for a brick $i\in[n]$, we lift the hyperplane to the pre-image under $L_i$ in $\R^s$ and compute the arrangement of the resulting hyperplanes in $\O_{ds\Delta}(n^s)$ time using the algorithm by Edelsbrunner et al. (\cref{theorem:edelsbrunner}). It is not hard to see that restricting $x_0$ to a face of this arrangement results in $L_i(x_0)$ living in a unique face of the $d$-dimensional arrangement for brick $i$ that corresponds to a domain for which we can obtain affine faithful decompositions. This is made explicit in the proof of \cref{theorem:uniform-4-block-xp}. A visualization of possible regions to which $x_0$ is restricted in the proof of \cref{theorem:uniform-4-block-xp} is shown in \cref{fig:domains}.

\begin{figure}
    \centering
    \includegraphics[width=0.8\linewidth]{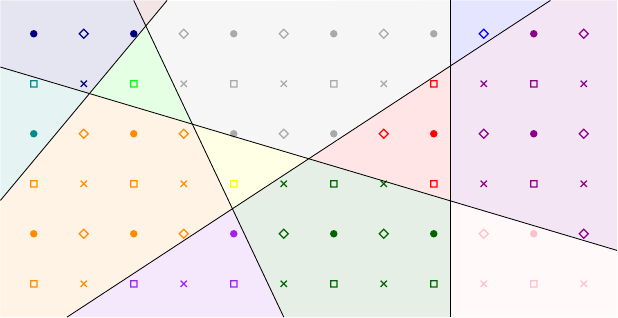}
    \caption{A partition of $\Z^2$ into four remainder classes modulo $2$, indicated by the differently shaped markers placed at each lattice point, and faces induced by $5$ hyperplanes, indicated by differently colored regions for the $2$-dimensional faces. Note that lower-dimensional faces also appear in arrangements and may contain integral points. To solve a 4-block ILP, we solve a configuration ILP for each set of identically shaped and colored points and restricting $x_0$ to be in such set.}
    \label{fig:domains}
\end{figure}

We are now ready to prove \cref{theorem:uniform-4-block-xp}.

\theoremuniformfourblockxp*

\begin{proofof*}{theorem:uniform-4-block-xp}
    Let $M=2^{\O(d\Delta)^{d^3}}$ and $\Xi=2^{\O(d\Delta)^d}$ be the numbers defined in \cref{lemma:dynamic-decomposition}. Let $\eta=\O(d(\Delta+\Xi))^d=\O(d(\Delta+2^{\O(d\Delta)^d}))^d=2^{\O(d\Delta)^d}$ be the norm bound on $\underline x$ defined in \cref{lemma:solution-decomposition} and let $\delta=\O(d\Delta)^d$ be the $\infty$-norm bound on Graver basis elements of any matrix $D\in\{-\Delta,-\Delta+1,\dots,\Delta\}^{d\times t}$ from~\cite{DBLP:conf/icalp/EisenbrandHK18}. Note that all mentioned numbers only depend on $d$ and $\Delta$.

    First, we apply \cref{lemma:dynamic-decomposition} to generate a collection of hyperplanes $H_i$ in $\R^d$ for each brick $i\in[n]$. Recall that these define the polyhedral regions in which the affine faithful decompositions of $b_i$ are valid. Lift each hyperplane to the space $\R^s$ through $L_i^{-1}(h):=(-C_i^\top a_h,\beta_h-a_h^\top b_i)$. We consider the arrangement induced by
    \[
        H=\bigcup_{i\in[n],h\in H_i}L_i^{-1}(h)
    \]
    and compute this arrangement $\mathcal F_H$ using \cref{theorem:edelsbrunner}. Observe that $|H|\le n\cdot\O(d\Delta)^{d^4}$. Using that $|\mathcal F_H|\le\O(|H|)^s$~\cite{DBLP:series/eatcs/Edelsbrunner87,grunbaum1967convex}, it follows that $|\mathcal F_H|\le\O(d\Delta)^{d^4s}\cdot n^s$. 
    
    We guess the nonempty face $F\in\mathcal F_H$ that contains $x_0$, i.e., the position $\rel$ of $x_0$ w.r.t.\ the hyperplanes $H$, and the remainder vector $r\in\{0,1,\cdots,M-1\}^s$ of $x_0$ modulo $M$. More concretely, for each combination of $r$ and $F\ne\emptyset$ among the 
    \[
        M^s\cdot|\mathcal F_H|\le2^{\O(d\Delta)^{d^3}\cdot s}\cdot\O(d\Delta)^{d^4s}\cdot n^s=2^{\O(d\Delta)^{d^3}\cdot s}\cdot n^s
    \]
    options, we solve a high-multiplicity encoding of the 4-block ILP using \cref{lemma:dynamic-decomposition}. In order to construct this encoding, for every brick $i\in[n]$, we compute:
    \begin{itemize}
        \item the face $F_i\in\mathcal F_{H_i}$ that contains $L_i(F)$, i.e., contains all $b_i-C_ix_0$ for all $x_0\in F$. This face has the position vector $\rel_i\in\signs^{H_i}$ given by $(\rel_i)_h=\rel_{L_i^{-1}(h)}$ for $h\in H_i$. This follows from the fact that for any $x_0$ it holds that
        \begin{alignat*}{2}
            &L_i(x_0)\mathrel{(\rel_i)_h}\beta_h&\iff& a_h^\top(b_i-C_ix_0)\mathrel{(\rel_i)_h}\beta_h\\
            \iff&(-C_i^\top a_h)^\top x_0\mathrel{(\rel_i)_h}\beta_h-a_h^\top b_i&\iff& a_{L_i^{-1}(h)}^\top x_0\mathrel{(\rel_i)_h}\beta_{L_i^{-1}(h)}.
        \end{alignat*}
        \item the remainder vector $r_i=(b_i-C_ir)\bmod M$ so that $L_i(x_0)\equiv r_i\pmod M$ when $x_0\equiv r\pmod M$.
        \item a dynamic faithful decomposition $\mathcal D_i$ using \cref{lemma:dynamic-decomposition} for the remainder vector $r_i$ and face $F_i$. That is, construct the affine map $\mathcal D_i\colon\R^d\to\R^{R_i}$ so that $\mathcal D_i(b)$ is a faithful decomposition of $b$ using elements from $R_i\subseteq\db d\Xi$ as long as $b\in(r_i+M\Z^d)\cap F_i$. By construction, this holds for $b=L_i(x_0)=b_i-C_ix_0$ when $x_0\in(r+M\Z^s)\cap F$. Represent $\mathcal D_i$ with a matrix in $\Q^{R_i\times d}$ and vector in $\Q^{R_i}$.
        \item the set of base solutions $\underline X_{b',i}=\{x\in\Z_{\ge0}^t:\|x\|_\infty\le\eta,D_ix=b'\}$ to $D_ix=b'$ for each small right-hand side vector $b'\in R_i$.
        \item the set of Graver basis elements $G_i:=\graverbasis(D_i)\cap\Z_{\ge0}^t$ of $D_i$\footnote{We note that the diagonal blocks $D_i$ in a 4-block ILP can be transformed to have nonnegative coefficients by introducing auxiliary local variables that model $x_i'=u-x_i$ for some upper bounds $u$ (enforced by local constraints). Then, if the local variables that do not appear in local constraints are treated separately, this results in $G_i=\emptyset$, which simplifies \cref{ilp:decomposed-4-block}, but comes at the cost of increasing the dependence of $t$ on the running time.}.
    \end{itemize}
    We can now solve the 4-block ILP restricted to $x_0\in(r+M\Z^s)\cap F$ by solving the following configuration ILP:
    \begin{alignat}{3}
       \text{minimize} \quad & \multicolumn{4}{l}{\normalsize$\displaystyle c_0^\top x_0+\sum_{i\in[n]}c_i^\top\bigl(\sum_{b'\in R_i}\sum_{\ x'\in\underline X_{b',i}}x'\cdot p_{b',x'}^i+\sum_{g\in G_i}g\cdot p_g^i\bigr),$} \label[ilp]{ilp:decomposed-4-block}\tag{C}\\
       \text{subject to} \quad & \multicolumn{4}{l}{\normalsize$\displaystyle Ax_0+B\Bigl(\sum_{x'\in\db t\eta\cap\Z_{\ge0}^t}x'\cdot q_{x'}+\sum_{g\in\db t\delta\cap\Z_{\ge0}^t}g\cdot q_g\Bigr)=b_0,$} \label[constraint]{constraint:global}\\
       & & q_{x'} & =\sum_{i\in[n]}\sum_{\ \substack{b'\in R_i:\\x'\in\underline X_{b',i}}}p_{b',x'}^i, & x' & \in\db t\eta\cap\Z_{\ge0}^t, \label[constraint]{constraint:aggregate-base-solutions}\\
       & & q_g & =\sum_{\substack{i\in[n]:\\g\in G_i}}p_g^i, & g & \in\db t\delta\cap\Z_{\ge0}^t, \label[constraint]{constraint:aggregate-graver}\\
       & \multicolumn{3}{l}{\normalsize$\displaystyle\sum_{x'\in\underline X_{b',i}}p_{b',x'}^i=(\mathcal D_i(b_i-C_ix_0))_{b'}$}, & i & \in[n],\,b'\in R_i, \label[constraint]{constraint:faithful-decomposition}\\
       & & q_{x'} & \in\Z_{\ge0}, & x' & \in\db t\eta\cap\Z_{\ge0}^t,\notag\\
       & & q_g & \in\Z_{\ge0}, & g & \in\db t\delta\cap\Z_{\ge0}^t,\notag\\
       & & p_{b',x'}^i & \in\Z_{\ge0}, & i\in[n],\,b' & \in R_i,\,x'\in\underline X_{b',i},\quad\quad\notag\\
       & & p_g^i & \in\Z_{\ge0}, & i & \in[n],\,g\in G_i,\notag\\
       & & x_0 & \in(r+M\Z^s)\cap F,\ x_0\ge0. & & \notag
    \end{alignat}
    
    By taking the optimal solution among all guesses of $r$ and $F$, we can report the optimal solution to the original 4-block ILP.

    The correctness of the algorithm follows directly from \cref{claim:correct-encoding} and the fact that the domains of the form $(r+M\Z^s)\cap F$ partition $\Z^s$.

    \begin{claimin}
        \cref{ilp:decomposed-4-block} correctly encodes \cref{ilp:4-block} subject to $x_0\in(r+M\Z^s)\cap F$.
        \label{claim:correct-encoding}
    \end{claimin}

    \begin{claimproof}
        As $b_i-C_ix_0\in F_i$ and $b_i-C_ix_0\in r_i+M\Z^d$ for all $x_0\in(r+M\Z^s)\cap F$, the faithful decompositions $\mathcal D_i(b_i-C_ix_0)$ are proper faithful decompositions of $b_i-C_ix_0$. Therefore, any solution $x_i\in\Z_{\ge0}^t$ to $D_ix_i=b_i-C_ix_0$ can be written as the sum
        \[
            x_i=\sum_{b'\in R_i}\sum_{\ x'\in X_{b'}^{D_i}}(u_{b'})_{x'}\cdot x'
        \]
        for some multisets $u_{b'}\in\Z_{\ge0}^{X_{b'}^{D_i}}$ of solutions to the decomposed blocks that matches the decomposition through $\sum_{x'\in X_{b'}^{D_i}}(u_{b'})_{x'}=(\mathcal D_i(b_i-C_ix_0))_{b'}$ for $b'\in R_i$. \cref{lemma:solution-decomposition} shows that any $x'\in X_{b'}^{D_i}$ can be written as $x'=x''+g_1+\dots+g_\ell$ for $x''\in\underline X_{b',i}$ and $g_1,\dots,g_\ell\in G_i$. Since adding any Graver basis element to $x_i$ still yields a solution to $D_ix_i=b_i-C_ix_0$, this justifies the encoding
        \[
            x_i=\sum_{b'\in R_i}\sum_{\ x'\in\underline X_{b',i}}p_{b',x'}^i\cdot x'+\sum_{g\in G_i}p_g^i\cdot g
        \]
        under \cref{constraint:faithful-decomposition} and arbitrary nonnegative multiplicities $p_g^i$.
    \end{claimproof}

    We now argue that \cref{ilp:decomposed-4-block} can be solved in FPT time.

    \begin{claimin}
        \cref{ilp:decomposed-4-block} can be solved in time $s^{\O(d\Delta)^d\cdot st+2^{\O(d\Delta)^d\cdot t}}\cdot m\cdot(nL)^{\O(1)}$.
    \end{claimin}

    \begin{claimproof}
        The constraint $x_0\in F$ is defined by a partially open polyhedron $F$ and its associated position vector $\rel\in\signs^H$. This corresponds to $a_h^\top x_0\mathrel{\rel_h}\beta_h$ for all $h\in H$. After rescaling the hyperplane coefficients to be integral, a strict inequality $a^\top x_0<\beta$ is equivalent to $a^\top x_0\le\beta-1$. Note that the condition $x_0\in r+M\Z^s$ can be enforced by adding the variables $v\in\Z^s$ and substituting $x_0=r+Mv$ throughout the formulation. This shows that \cref{ilp:decomposed-4-block} corresponds to an integer program with linear constraints.
        
        We now estimate the dimensions of \cref{ilp:decomposed-4-block}. When we aggregate the $v$ and $q$ variables into $z$, the ILP has the form $\min\{c^{\prime\top}(z,p)\ \vert\ Q_1z+Q_2p=f_1,Q_3z\le f_2,(z,p)\in\Z^{N_\Z}\times\Z_{\ge0}^{N_\R}\}$, where $N_\Z\le s+(\eta+1)^t+(\delta+1)^t=s+2^{\O(d\Delta)^d\cdot t}$ and $N_{\R}\le n\cdot((d+1)\cdot(\eta+1)^t+(\delta+1)^t)=2^{\O(d\Delta)^d\cdot t}\cdot n$. The number of constraints of \cref{ilp:decomposed-4-block} is $\O(|H|+m+N_\Z+N_\R)=\O(m+N_{\Z}+N_\R)$. We now estimate the largest absolute value of the numerator and denominator of any number that appears in the ILP after substituting $x_0=r+Mv$.
        \begin{itemize}
            \item The objective coefficients and coefficients appearing in \cref{constraint:global} are integers bounded by $\O(2^L)\cdot t\cdot(\eta+\delta)\cdot M=2^L\cdot t\cdot 2^{\O(d\Delta)^{d^3}}$.
            \item The matrix $\overline W$ used in the proof of \cref{lemma:dynamic-decomposition} has numerators and denominators bounded by $\Psi\cdot\Phi^dd^{d/2}=2^{\O(d\Delta)^d}$. Therefore, the representation of $\alpha,\gamma$ and $z$ in the proof of \cref{lemma:dynamic-decomposition}, which make up $\mathcal D_i$, have numerators and denominators bounded by the same asymptotic quantity. As a result, the numerator and denominators of the coefficients appearing on the right-hand side of \cref{constraint:faithful-decomposition} are bounded by $2^L\cdot2^{\O(d\Delta)^d}\cdot M=2^L\cdot2^{\O(d\Delta)^{d^3}}$.
            \item The normal vectors $a_h$ of the hyperplanes in $H_i$ correspond to the rows of matrices~$\overline V^{-1}$ in the proof of \cref{lemma:dynamic-decomposition}. By multiplying with $\det\overline V$ we obtain integral hyperplane normal vectors bounded by $\Phi^dd^{d/2}=\O(d\Delta)^{d^3}$. The hyperplane offsets $\beta_h$ become integers bounded by $\det\overline V\cdot(t\cdot\Psi)=2^{\O(d\Delta)^d}\cdot t$. Following the definition of $L_i^{-1}(h)$, $F$ can be represented with integral coefficients that are bounded by $2^L\cdot t\cdot2^{\O(d\Delta)^d}\cdot M=2^L\cdot t\cdot2^{\O(d\Delta)^{d^3}}$.
        \end{itemize}
        Taking the logarithm shows that the encoding lengths of all rational numbers in the instance are bounded by $L'=\log(2^L\cdot t\cdot2^{\O(d\Delta)^{d^3}})=L+\log t+\O(d\Delta)^{d^3}$.
    
        We now repeat the observation by Cslovjecksek et al.~\cite{DBLP:journals/theoretics/CslovjecsekKLPP25}, which reveals that \cref{ilp:decomposed-4-block} can be solved efficiently. Since every variable $p_{b',x'}^i$ appears in one constraint of type (\ref{constraint:faithful-decomposition}) and one constraint of type (\ref{constraint:aggregate-base-solutions}) and every variable $p_g^i$ appears in one constraint of type (\ref{constraint:aggregate-graver}), the matrix $Q_2$ is totally unimodular by Ghouila-Houri's criterion. Therefore, the ILP can be solved by solving the mixed-integer relaxation $\min\{c^{\prime\top}(z,p)\ \vert\ Q_1z+Q_2p=f_1,Q_3z\le f_2,(z,p)\in\Z^{N_\Z}\times\R_{\ge0}^{N_\R}\}$. Using Kannan's algorithm~\cite{DBLP:journals/mor/Kannan87}, we can solve this MILP in time
        \begin{align*}
            &N_\Z^{\O(N_\Z)}\cdot\O(m+N_\Z+N_\R)\cdot(N_\R L')^{\O(1)}\\
            &=(s+2^{\O(d\Delta)^d\cdot t})^{\O(s+2^{\O(d\Delta)^d\cdot t})}\cdot m\cdot(2^{\O(d\Delta)^d\cdot t}\cdot n\cdot (L+\log t+\O(d\Delta)^{d^3}))^{\O(1)}\\
            &=s^{\O(d\Delta)^d\cdot st+2^{\O(d\Delta)^d\cdot t}}\cdot m\cdot(nL)^{\O(1)}.
        \end{align*}
    \end{claimproof}

    We now estimate the running time of the other parts of algorithm. We can employ \cref{lemma:dynamic-decomposition} to compute the hyperplanes $H$ in $\O(d\Delta)^{d^4}\cdot n$ time and use \cref{theorem:edelsbrunner} to compute the arrangement $\mathcal F_H$ in $\O(n\cdot\O(d\Delta)^{d^4})^s=\O(d\Delta)^{d^4s}\cdot n^s$ time. For each of $2^{\O(d\Delta)^{d^3}\cdot s}\cdot n^s$ many domain guesses of $(r+M\Z^s)\cap F$ and for each brick $i\in[n]$, we can derive $F_i$ and $r_i$ in a negligible amount of time. Extracting $\mathcal D_i$ with \cref{lemma:dynamic-decomposition} takes $\O(\Delta)^{d^4}$ time. Since $G_i\subseteq\{0,1,2,\dots,\delta\}^t$, we can naively compute $G_i:=\graverbasis(D_i)\cap\Z_{\ge0}^t$ in $(dt)^{\O(1)}\cdot((\delta+1)^t)^2=(d\Delta)^{\O(dt)}$ time. Similarly, for each $b'\in R_i$, we can compute $\underline X_{b',i}$ in time $(dt)^{\O(1)}\cdot(\eta+1)^t=2^{\O(d\Delta)^d\cdot t}$.

    Combining all previous estimations shows that the running time of the entire algorithm is bounded by
    \begin{align*}
        &\O(d\Delta)^{d^4}\cdot n+\O(d\Delta)^{d^4s}\cdot n^s+2^{\O(d\Delta)^{d^3}\cdot s}\cdot n^s\\
        &\cdot(n\cdot(\O(\Delta)^{d^4}+(d\Delta)^{\O(dt)}+(d+1)\cdot2^{\O(d\Delta)^d\cdot t})+s^{\O(d\Delta)^d\cdot st+2^{\O(d\Delta)^d\cdot t}}\cdot m\cdot(nL)^{\O(1)})\\
        &=2^{\O(d\Delta)^{d^3}\cdot s}\cdot s^{\O(d\Delta)^d\cdot st+2^{\O(d\Delta)^d\cdot t}}\cdot n^s\cdot m\cdot(nL)^{\O(1)}.
    \end{align*}
\end{proofof*}

Similarly to the $n$-fold algorithm by Cslovjecksek et al.~\cite{DBLP:journals/theoretics/CslovjecsekKLPP25}, the running time in \cref{theorem:uniform-4-block-xp} only depends polynomially on the number of global constraints\footnote{This is not surprising: $B$-uniformity implies that the column rank of the first $m$ rows of the 4-block constraint matrix is at most $s+t$ and we may assume that $m\le s+t$ after removing redundant constraints. On the other hand, it is notable that the techniques used in the proof of \cref{theorem:uniform-4-block-xp} still yield an algorithm that has a running time that is linear in $m$ when the global constraints of \cref{ilp:4-block} are formulated in inequality form as $Ax_0+B(x_1+\dots+x_n)\le b_0$.}.

We now discuss how to obtain \cref{theorem:4-block-strongly-xp} from the proof of \cref{theorem:uniform-4-block-xp}. The first step is to reduce a non-uniform 4-block ILP to a $B$-uniform problem. Dvořák et al.~\cite{DBLP:journals/ai/DvorakEGKO21} provide a parameterized reduction that reduces non-uniform 4-block ILPs to uniform 4-block ILPs. Since they aim for $D_1=\dots=D_n$, their reduction increases $d$ exponentially. Therefore, naively implementing their reduction would result in a quadruply exponential parameter dependence hidden in the constant in \cref{theorem:4-block-strongly-xp}. For this reason, we explicitly show in \cref{lemma:reduction-to-uniform} how non-uniform 4-block ILP may be reduced to $B$-uniform 4-block ILP while not increasing $d$ significantly.

\begin{lemma}
    4-block ILP is reducible to $B$-uniform 4-block ILP by changing $t$ to $(2\overline\Delta+1)^{m+d}$ and additively increasing $d$ by $1$.
    \label{lemma:reduction-to-uniform}
\end{lemma}

\begin{proof}
    Given an non-uniform 4-block ILP instance, we construct a new $B$-uniform 4-block ILP with an equal number of bricks. We use primes to denote parts of the constructed instance, which will have $t'=(2\overline\Delta+1)^{m+d}$ local variables and $d'=d+1$ local constraints per brick.

    Associate every vector $(p,q)\in\db m{\overline\Delta}\times\db d{\overline\Delta}$ with a unique index $j\in[t']$. For a brick $i\in[n]$, we construct the $j$-th column of the new instance, $j\in[t']$ corresponding to $(p,q)$, by uniformly setting the coefficients $(B_i')_j=p$. We set $(c_i')_j=(c_i)_{j^*}$ and $(D_i')_j=\binom0q$ if the constraint matrix coefficients $(p,q)$ occur in the original brick and $j^*\in[t]$ indexes the local variable with minimum objective coefficient $(c_i)_{j^*}$ among the local variables $(x_i)_j$ with $(B_i)_j=p,(D_i)_j=q$. Otherwise, if there is no such $j^*$, we set $(c_i')_j=0$ and $(D_i')_j=\binom1q$. The rest of the new brick is constructed by setting $C_i'=\binom{0^\top}{C_i}$ and $b_i'=\binom0{b_i}$. The top left block $A'=A$ is copied.

    The additional local constraint in the construction enforces that $(x_i)_j=0$ if the corresponding column $(p,q)$ does not appear in brick $i$ in the original instance. In this way, the constructed instance corresponds to a reduced version of the original ILP where duplicated columns with worse objective coefficients are removed, which, in turn, is equivalent to the original ILP.
\end{proof}

We note that the construction in \cref{lemma:reduction-to-uniform} is not limited to small coefficients. In particular, when $k$ is the number of distinct coefficients occurring in any $B_i$, we may assume that $t'\le(2\Delta+1)^d\cdot k^m$. This yields a $\O_{dms\Delta k}(n^{s+\O(1)}\cdot t\cdot L^{\O(1)})$ time algorithm for solving 4-block ILPs.

We now prove \cref{theorem:4-block-strongly-xp} by following Step 1 and 4 of the strongly polynomial framework by Koutecký, Levin and Onn~\cite{DBLP:conf/icalp/KouteckyLO18}.

\theoremfourblockstronglyxp*

\begin{proof}
    As a result of \cref{lemma:reduction-to-uniform}, it suffices to restrict ourselves to solving $B$-uniform 4-block ILPs with $t=\O_{dm\soverline\Delta}(1)$ after a preprocessing step using $\O_{dms\soverline\Delta}(nt)$ time. We solve such instance using the same algorithm as that of \cref{theorem:uniform-4-block-xp} while taking care to solve the instances of \cref{ilp:decomposed-4-block} in strongly polynomial time. This is possible because all coefficients are now bounded by $\O_{dms\soverline\Delta}(1)$. To make this explicit, \cref{ilp:decomposed-4-block} can be solved in strongly polynomial time with the following steps where $N=\O_{dms\soverline\Delta}(n)$ is the number of variables of \cref{ilp:decomposed-4-block} and $\Box=\O_{dm\soverline\Delta}(1)$ is the largest absolute value of a coefficient in \cref{ilp:decomposed-4-block}:
    \begin{itemize}
        \item Solve the LP relaxation of \cref{ilp:decomposed-4-block} using a linear programming algorithm that is strongly polynomial when the constraint matrix entries are bounded, for instance, using Tardos' algorithm~\cite{DBLP:journals/ior/Tardos86}.
        \item If the LP relaxation is infeasible, the ILP is also infeasible. Otherwise, if the LP relaxation is unbounded, we are left with checking the feasibility of the ILP or, equivalently, optimizing with $c=0$. Therefore, we may assume that the LP relaxation is bounded.
        \item We can obtain an optimal fractional solution $(z^*,p^*)$ and add additional constraints restricting $(z,p)$ to a bounding box with sides bounded by $Q:=2(N\cdot\Box^NN^{N/2})+1$ using the proximity result by Cook et al.~\cite{DBLP:journals/mp/CookGST86}.
        \item Translate the resulting instance so that the variable lower bounds become zero\footnote{We note that, in order to compute the resulting translated box, we need to round $(z^*,p^*)$ to an integer point. Here, we assume that this can be done in a time that is linear in the number of coordinates, which is standard in the integer programming literature. A similar operation needs to be done when computing the remainder vectors $r_i$.}. The right-hand side vector of the ILP must now have an $\infty$-norm bounded by $\Box NQ$ or the ILP is trivially infeasible.
        \item Since any two solutions to the box-constrained ILP have a $1$-norm distance of at most $NQ$, we may apply Frank and Tardos' coefficient reduction technique~\cite{DBLP:journals/combinatorica/FrankT87} to replace the objective vector $c$ with an objective of $\infty$-norm at most $2^{4N^3}(NQ+1)^{N(N+2)}$ that has the same set of optimizers over the integral solutions to \cref{ilp:decomposed-4-block}.
        \item The encoding lengths of all the numbers in the resulting equivalent ILP instance are bounded by $\log((NQ)^{N^{\O(1)}})=N^{\O(1)}\cdot\log(NQ)=N^{\O(1)}\cdot\log\Box=\O_{dms\soverline\Delta}(n^{\O(1)})$. Therefore, solving the ILP with any (weak) FPT algorithm will give the required running time. Therefore, we again relax the $p$ variables to $\R^{N_\R}$ and solve the resulting MILP with Kannan's algorithm~\cite{DBLP:journals/mor/Kannan87}.
    \end{itemize}
\end{proof}

Note that due to the construction in \cref{lemma:reduction-to-uniform}, the constant hidden in the $\O_{dms\soverline\Delta}$ in \cref{theorem:4-block-strongly-xp} depends triply exponentially on the block dimensions $d$ and $m$.

\section{Dynamic Faithful Decomposition}
\label{sec:dynamic-faithful-decomposition}

This section proves \cref{lemma:dynamic-decomposition}.

\lemmadynamicdecomposition*

The essence of the proof of \cref{lemma:dynamic-decomposition} is showing that by restricting $b$ to a suitable face $F$ of a hyperplane arrangement and restricting it to a lattice translate $r+M\Z^d$, all necessary ingredients to inductively apply the proof of Theorem 4.1 of Cslovjecksek et al.~\cite{DBLP:conf/esa/CslovjecsekEPVW21} are known up front. For a vector $b\in(r+M\Z^d)\cap F$, the proof proceeds as follows:
\begin{itemize}
    \item We identify the collection $\mathcal U$ of bases $U\in\dbb d\Delta$ that conically generate $b$. This is determined by the face $F$.
    \item We identify a linearly independent set of vectors in $\bigcap_{U\in\mathcal U}\intcone U$ that conically generate $b$. This is similarly determined by $F$.
    \item If a conic multiplier in the corresponding conic combination is large, then $b$ faithfully decomposes into the corresponding generator $w_i$ and $b-w_i$. We exhaustively and inductively apply this. Here, the induction base case is represented by a unique root $q$. The root is determined by the guessed lattice translate $r+M\Z^d$ and the subset of conic multipliers that are sufficiently large, which is identified through $F$.
\end{itemize}

We note that performing work for each lattice translate and face of an arrangement has also recently been employed in~\cite{DBLP:conf/soda/BachERW25} to verify forall-exists statements over the integers.

As the induction step of \cref{claim:induction} in the proof of \cref{lemma:dynamic-decomposition} follows the proof of Cslovjecksek et al.~\cite{DBLP:conf/esa/CslovjecsekEPVW21}, we will also need their auxiliary lemma, \cref{lemma:intersection-integer-cone}, that reveals important properties of a cone of the form $\bigcap_{U\in\mathcal U}\cone U$ where $\mathcal U\subseteq\dbb d\Delta$. Here, Cslovjecksek et al.\ show that one can find a set of conic generators $P$ of bounded size that also belong to each of the integer cones $\intcone U$ for $U\in\mathcal U$.

\begin{lemma}[Lemma 4.2 in~\cite{DBLP:conf/esa/CslovjecsekEPVW21}]
    Let $d,\Delta\in\Z_{\ge0}$ be given. Then, there exists numbers $\Psi=2^{\O(d\Delta)^d}$ and $\Phi=\O(d\Delta)^{d^2}$ such that the following holds:  for any subset of generating bases $\mathcal U\subseteq\dbb d\Delta$ for which $\bigcap_{U\in\mathcal U}\cone U\ne\emptyset$, there exist a generating set $P\subseteq\db d\Phi$ of elements with bounded norm such that
    \[
        \cone P=\bigcap_{U\in\mathcal U}\cone U.
    \]
    Furthermore, for any $v\in\Z^d\cap\bigcap_{U\in\mathcal U}\cone U$ it holds that
    \[
        \Psi\cdot v\in\bigcap_{U\in\mathcal U}\intcone U.
    \]
    \label{lemma:intersection-integer-cone}
\end{lemma}

We are now ready to prove \cref{lemma:dynamic-decomposition}. The decomposition mechanism is visualized in \cref{fig:dynamic-decomposition}.

\begin{figure}[H]
    \begin{subfigure}[t]{0.45\textwidth}
        \centering
        \includegraphics[width=\textwidth]{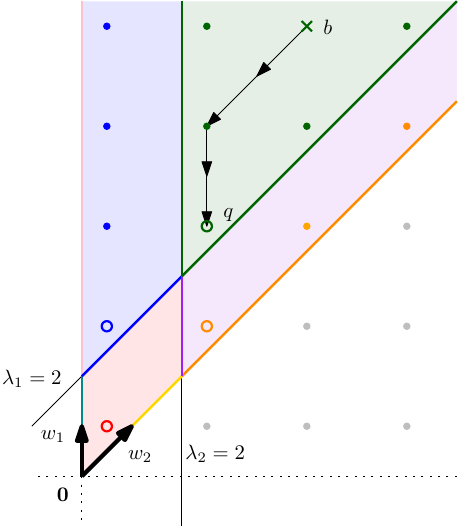}
        \caption{Lattice translate $(1,2)+4\Z^2$. Arrows indicate the dynamic decomposition of the right-hand side $b=(9,18)$ into $q=(5,10)$, $2$ times $w_1$ plus $2$ times $w_2$. Each arrow corresponds to an induction step of \cref{claim:induction}.}
    \end{subfigure}
    \hspace{0.05\textwidth}
    \begin{subfigure}[t]{0.45\textwidth}
        \centering
        \includegraphics[width=\textwidth]{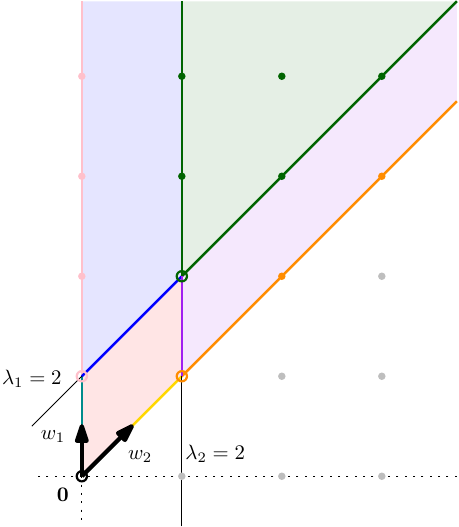}
        \caption{Lattice translate $(0,0)+4\Z^2$.}
    \end{subfigure}
    \caption{Visual representation of a subregion of $\R^2$ in the proof of \cref{lemma:dynamic-decomposition} for $d=2$ dimensions, $t=2$, and $\Delta=1$. The union of the blue, red, purple, and green region corresponds to the set of generating bases
    \[
        \mathcal U=\left\{\begin{pmatrix}1&0\\-1&1\end{pmatrix},\begin{pmatrix}1&0\\0&1\end{pmatrix},\begin{pmatrix}1&-1\\0&1\end{pmatrix},\begin{pmatrix}1&0\\1&1\end{pmatrix},\begin{pmatrix}1&-1\\1&1\end{pmatrix},\begin{pmatrix}1&-1\\1&0\end{pmatrix}\right\}.
    \]
    We have that $\bigcap_{U\in\mathcal U}\cone U$ is generated by $v_1=(0,1),v_2=(1,1)$ and that the multiplier $\Psi=2$ suffices to obtain the generating elements $w_1=2v_1=(0,2),w_2=2v_1=(2,2)$ that are in the intersection of integer cones. Furthermore, because $\det W=2$, it suffices to partition the space into lattice translates of $M\Z^2$ where $M=2\cdot2=4$. The region associated with $\mathcal U$ is divided into $4$ subregions based on whether $\lambda_1\ge2$ and/or $\lambda_2\ge2$. For each such subregion, the root lattice point $q$ is marked with a circle. Note that the (degenerate $1$-dimensional) cyan-pink and yellow-orange cones can each be generated by a strictly larger set of bases and are therefore assigned a different set of generating bases $\mathcal U$.}
    \label{fig:dynamic-decomposition}
\end{figure}

\begin{proofof}{lemma:dynamic-decomposition}
    First, let $\Phi=\O(d\Delta)^{d^2}$ be the norm upper bound on the generating elements $P$ defined in \cref{lemma:intersection-integer-cone} and let $\Psi=2^{\O(d\Delta)^d}$ be the large multiplier defined in \cref{lemma:intersection-integer-cone}. If necessary, we enlarge $\Phi$ such that $1\le\Delta\le\Phi$.

    Due to \cref{claim:standard-matrix}, we may assume that $D$ is full row rank and $t=\O(\Delta)^d$.

    \begin{claimin}
        A faithful decomposition of $b$ w.r.t.\ the matrix with column set $\db d\Delta$ is a faithful decomposition with respect to any $D\in\{-\Delta,-\Delta+1,\dots,\Delta\}^{d\times t}$.
        \label{claim:standard-matrix}
    \end{claimin}

    \begin{proof}
        It suffices to show that faithful decompositions, that are crafted with respect to modified matrices by adding columns or removing duplicate columns, are also faithful decompositions with respect to the original matrix. Let $D,D'$ be matrices with integral coefficients bounded by $\Delta$ in absolute value and let $b=b_1+\dots+b_\ell$ be a faithful decomposition of some $b\in\Z^d$ w.r.t.\ $D'$.
        
        First, consider the case where a matrix $D'\in\Z^{d\times(t+1)}$ arises from a matrix $D\in\Z^{d\times t}$ by adding a $(t+1)$-th column to $D$. If $x_0\in\Z_{\ge0}^t$ is a solution to $Dx_0=b$, then $x_0'\in\Z_{\ge0}^{t+1}$ obtained by padding $x_0$ with a zero in the $(t+1)$-th coordinate is a solution to $D'x_0'=b$, which decomposes as $x_0'=x_1'+\dots+x_\ell'$, where $D'x_i'=b_i,x_i'\in\Z_{\ge0}^{t+1}$. By letting $x_i\in\Z_{\ge0}^t$ be the projection of $x_i'$ on the first $t$ coordinates, it follows that $x_0=x_1+\dots+x_\ell$ with $Dx_i=b_i$.

        Second, consider the case where $D\in\Z^{d\times(t+1)}$ has identical $t$-th and $(t+1)$-th columns and $D'\in\Z^{d\times t}$ consists of the first $t$ columns of $D$. In this case, a solution $x_0\in\Z_{\ge0}^{t+1}$ to $Dx_0=b$ can be aggregated into a solution $x_0'\in\Z_{\ge0}^t$ to $D'x_0'=b$ by setting $(x_0)_j'=(x_0)_j$ for $j\in[t-1]$ and $(x_0')_t=(x_0)_t+(x_0)_{t+1}$. Consider solutions $x_i'\in\Z_{\ge0}^t$ in a decomposition $x_0'=x_1'+\dots+x_\ell'$, where $D'x_i'=b_i$. We reverse the aggregation operation by copying $(x_i)_j=(x_i')_j$ for $j\in[t-1]$ and setting the other two components through
        \[
            \begin{pmatrix}(x_i)_t\\(x_i)_{t+1}\end{pmatrix}=
            \begin{cases}
                \begin{pmatrix}(x_i')_t\\0\end{pmatrix},&\text{if $\sum_{j=1}^i(x_j')_t\le x_t$,}\\[10pt]
                \begin{pmatrix}x_t-\sum_{j=1}^{i-1}(x_j')_t\\\sum_{j=1}^i(x_j')_t-x_t\end{pmatrix},&\text{if $\sum_{j=1}^{i-1}(x_j')_t\le x_t<\sum_{j=1}^i(x_j')_t$,}\\[10pt]
                \begin{pmatrix}0\\(x_i')_t\end{pmatrix},&\text{if $x_t<\sum_{j=1}^{i-1}(x_j')_t$.}
            \end{cases}
        \]
        By construction, the vectors $x_i\in\Z_{\ge0}^{t+1}$ are solutions to $Dx_i=b_i$. Therefore, we obtain the decomposition $x_0=x_1+\dots+x_\ell$.
    \end{proof}
    
    For every basis $\overline V\in\dbb d\Phi$ of vectors in $\db d\Phi$, we add $2d$ hyperplanes to $H$. In particular, for each $i\in[d]$, we add the two hyperplanes
    \[
        h_{\overline Vi}^0=((\soverline V^{-1}b)_i,0)\text{ and }h_{\overline Vi}^+=((\soverline V^{-1}b)_i,t\cdot\Psi)
    \]
    to $H$. We have that
    \[
        |H|\le2d\cdot\binom{|\db d\Phi|}d\le2d\cdot|\db d\Phi|^d=2d\cdot\left(\left(2\cdot\O(d\Delta)^{d^2}+1\right)^d\right)^d=\O(d\Delta)^{d^4}.
    \]
    
    Let us fix an arbitrary face $F\in\mathcal F_H$ with position vector $\rel\in\signs^H$. Before describing the construction of $\mathcal D_{r,F}$, we prove some useful properties that hold uniformly for all vectors $b$ in the face.
    
    \begin{claimin}
        We can identify a set of generating bases $\mathcal U\subseteq\dbb d\Delta$ so that $F\subseteq\bigcap_{U\in\mathcal U}\cone U$. Furthermore, for all $b\in F$ and $U\in\dbb d\Delta$, we have that $b\in\cone U$ implies that $U\in\mathcal U$.
        \label{claim:generating-bases}
    \end{claimin}

    \begin{claimproof}
        Since $\Delta\le\Phi$, we have that any generating basis $U\in\dbb d\Delta$ appears as $\overline V=U$. We take $\mathcal U$ to be the set of all $U\in\dbb d\Delta$ so that $\rel_{h_{Ui}^0}\in\{>,=\}$ for all $i\in[d]$, i.e., $U^{-1}F\ge0$. This immediately implies that $F\subseteq\cone U$ for all $U\in\mathcal U$. Furthermore, if $b\in\cone U$ for some $b\in F$, then $U^{-1}b\ge0$ and thus $U^{-1}F\ge0$, due to the uniformity of the orientation of the face $F$ w.r.t.\ the hyperplanes. This shows that $U\in\mathcal U$.
    \end{claimproof}

    \begin{claimin}
        We can find a set of linearly independent generating vectors $V=\{v_1,\dots,v_\ell\}\subseteq\db d\Phi$ such that $F\subseteq\cone V\subseteq\bigcap_{U\in\mathcal U}\cone U$.
        \label{claim:big-generators}
    \end{claimin}

    \begin{claimproof}
        Using \cref{lemma:intersection-integer-cone}, there exists a generating set $P\subseteq\db d\Phi$ such that $F\subseteq\cone P=\bigcap_{U\in\mathcal U}\cone U$. Let $b$ be an arbitrary vector in $F$. By Caratheodory's theorem, there exists a linearly independent set $V=\{v_1,\dots,v_\ell\}\subseteq P\subseteq\db d\Phi$ so that $b\in\cone V$, i.e., there are strictly positive multipliers $(\mu_i)_{i\in[\ell]}$ so that
        \[
            b=\sum_{i\in[\ell]}\mu_i\cdot v_i.
        \]
        By adding unit vectors, we can extend $V$ to a basis $\overline V$ of vectors in $\db d\Phi$. As $b=\overline V(\soverline V^{-1})b=V\mu$, we have that $(\soverline V^{-1}b)_i=\mu_i>0$ for $i\in[\ell]$ and $(\soverline V^{-1}b)_i=0$ for $i\in[d]\setminus[\ell]$. As the sign of these coefficients $(\soverline V^{-1}b)_i$ is determined uniformly in $F$, we have that any $b\in F$ can written as a conic combination of elements in $V$ and, consequently, that $F\subseteq\cone V$.
    \end{claimproof}

    Let $w_i=\Psi\cdot v_i$ for $i\in[\ell]$. It will be of interest to inspect the scaled conic multipliers $\lambda_i:=\mu_i/\Psi>0$ in the conic combination $b=\sum_{i\in[\ell]}\lambda_i\cdot w_i$. Let $r\in\{0,1,\dots,M-1\}^d$ be a fixed remainder vector. We will pick a suitable modulus in the next argument.

    \begin{claimin}
        The remainder of $\alpha_i:=\lambda_i\bmod t$ of the conic multipliers $\lambda_i$ for $i\in[\ell]$ is constant for all $b\in(r+M\Z^d)\cap F$ and can be computed.
        \label{claim:fixed-remainder}
    \end{claimin}

    \begin{claimproof}
        Set $\overline W=\Psi\cdot\overline V$, i.e., the first $\ell$ vectors of $\overline W$ are $w_1,\dots,w_\ell$ and the last vectors are the $\Psi$-scaled dummy vectors to obtain a basis. It is known that $\soverline W^{-1}=\tfrac1\Psi\cdot\soverline V^{-1}\in\tfrac1{\Psi\cdot\det\overline V}\cdot\Z^{d\times d}$. We may now make a similar observation to that made in~\cite{DBLP:conf/esa/CslovjecsekEPVW21} which the authors use to derive \cref{lemma:intersection-integer-cone}. As $|\det\overline V|\le\Phi^dd^{d/2}$ by the Hadamard bound, we may observe that if we take $M=t\cdot\Psi\cdot\lcm\{1,2,\dots,\Phi^dd^{d/2}\}$, it holds that $M\cdot\soverline W^{-1}\in t\Z^{d\times d}$ independently of $\overline W$. This modulus is bounded by $M\le t\cdot\Psi\cdot3^{\Phi^dd^{d/2}}=\O(\Delta)^d\cdot2^{\O(d\Delta)^d}\cdot3^{(\O(d\Delta)^{d^2})^dd^{d/2}}=2^{\O(d\Delta)^{d^3}}$~\cite{Hanson_1972}.

        Therefore, if $r\in\{0,1,\dots,M-1\}^d$ is a fixed remainder vector and $b$ is some vector in the lattice translate $r+M\Z^d$, it holds that $\soverline W^{-1}b\in\soverline W^{-1}(r+M\Z^d)\subseteq\soverline W^{-1}r+(M\cdot\soverline W^{-1})\Z^d\subseteq\soverline W^{-1}r+t\Z^d$. Thus, for all $i\in[\ell]$, the remainder of $\lambda_i$ modulo $t$ is known to be fixed to $\lambda_i\bmod t=(\soverline W^{-1}r)_i\bmod t$. Note that this remainder may be fractional.
    \end{claimproof}

    Since we also know the position of $F$ relative to the hyperplanes $h_{\overline Vi}^+$, i.e., it is known whether $(\soverline V^{-1}b)_i\ge t\cdot\Psi\iff(\soverline W^{-1}b)_i\ge t$, we can distinguish the following two cases for every $i\in[\ell]$, uniformly over all vectors $b$ in the domain $(r+M\Z^d)\cap F$: either $\lambda_i=\alpha_i$; or $\lambda_i\ge t$ and, if we define $\gamma_i=t+\alpha_i$, the difference $z_i=\lambda_i-\gamma_i\in t\Z_{\ge0}$ is a nonnegative integer.
    
    Let $S\subseteq[\ell]$ be the set of directions $i\in[\ell]$ for which $\lambda_i=\alpha_i$, i.e., $F$ is not deep in $\cone V$ in the direction of $w_i$. Let $q=\sum_{i\in S}\alpha_i\cdot w_i+\sum_{i\in[\ell]\setminus S}\gamma_i\cdot w_i$, which, as argued, only depends on $r$ and $F$, but not on $b$. We are now ready to define the dynamic faithful decomposition $\mathcal D_{r,F}:\R^d\to\R^{R_{r,F}}$ by
    \[
        \mathcal D_{r,F}(b)=1\cdot e_q+\sum_{i\in[\ell]\setminus S}z_i\cdot e_{w_i},
    \]    
    where $e_{b'}$ is the $b'$-th standard basis vector, i.e., $e_{b'}(b')=1$ and $e_{b'}(b'')=0$ for $b''\ne b'$. That is, we decompose $b$ into $z_i$ times $w_i$ for $i\in[\ell]\setminus S$ and one times the root $q$. Note that $\|w_i\|_\infty=\Psi\cdot\|v_i\|\le\Psi\cdot\Phi=2^{\O(d\Delta)^d}\cdot\O(d\Delta)^{d^2}=2^{\O(d\Delta)^d}$ for all $i\in[\ell]$ and that
    \[
        \|q\|_\infty\le\sum_{i\in S}\alpha_i\|w_i\|_\infty+\sum_{i\in[\ell]\setminus S}\gamma_i\|w_i\|_\infty\le\sum_{i\in[\ell]}2t\|w_i\|_\infty\le d\cdot2t\cdot2^{\O(d\Delta)^d}=2^{\O(d\Delta)^d}.
    \]
    Therefore, the right-hand side vectors in the faithful decomposition $\mathcal D_{r,F}$ are within $R_{r,F}=\{q,w_1,\dots,w_\ell\}\subseteq\db d\Xi$ for $\Xi=2^{\O(d\Delta)^d}$ and $|R_{r,F}|\le d+1$. Note that the function is affine in $b$ as $z_i=\lambda_i-\gamma_i=(\soverline W^{-1}b)_i-\gamma_i$.

    From now on, let $b\in(r+M\Z^d)\cap F$ be fixed.
    Observe that the generating basis $U\in\dbb d\Delta$ which has $i$-th column $e_i$ if $b_i\ge0$ and $-e_i$ if $b_i<0$ for $i\in[d]$ is included in $\mathcal U$, as it generates an orthant that contains $b$. Therefore, all generators $w_i$ are sign-compatible with $b$, showing that the decomposition into elements of $R_{r,F}$ is sign-compatible.

    It is left to show that the decomposition matches the remaining properties in \cref{def:faithful-decomposition}. We prove this with induction on decreasing $\|z\|_1\in\Z_{\ge0}$ in \cref{claim:induction}. For this, we need the result from \cref{claim:stay-in-the-cone}.

    \begin{claimin}
        Let $\nu_i\in\R_{>0}^\ell$ be strictly positive multipliers and $p=\sum_{i\in[\ell]}\nu_i\cdot w_i$. Let $U\in\dbb d\Delta$ and suppose that $p\in\cone U$. Then, it holds that $U\in\mathcal U$.
        \label{claim:stay-in-the-cone}
    \end{claimin}

    \begin{claimproof}
        Suppose for a contradiction that $b\notin\cone U$. Let $Y\ne\emptyset$ be the set of indices $j\in[d]$ such that $(U^{-1}b)_j<0$ and construct another distinct basis $U'\ne U$ by flipping the sign of the basis vectors indexed by $Y$. In this way, $b\in\cone U'$. If there exists an index $j\in Y$ so that also $(U^{-1}p)_j>0$, we find that $p\notin\cone U'$. By using \cref{claim:generating-bases} and $b\in\cone U'$, we find that $p\in\cone V\subseteq\bigcap_{U''\in\mathcal U}\cone U''\subseteq\cone U'$, which is a contradiction. Therefore, for all $i\in Y$, it must hold that $(U'^{-1}p)_j=0$ and $(U'^{-1}b)_j>0$. Since $w_i\in\cone V\subseteq\cone U'$ for all $i\in[\ell]$, we have that $(U'^{-1}w_i)\ge0$. Combining this with $\nu_i>0$ for all $i\in[\ell]$ and
        \[
            \sum_{i\in[\ell]}\nu_i\cdot\sum_{j\in Y}(U'^{-1}w_i)_j=\sum_{j\in Y}\Bigl(U'^{-1}\sum_{i\in[\ell]}\nu_i\cdot w_i\Bigr)_j=\sum_{j\in Y}(U'^{-1}b)_j>0
        \]
        shows that there must be an $i\in[\ell]$ and $j\in Y$ such that $(U'^{-1}w_i)_j>0$. However, then it follows that
        \[
            (U'^{-1}p)_j=\Bigl(U'^{-1}\sum_{i\in[\ell]}\nu_i\cdot w_i\Bigr)_j=\sum_{i\in[\ell]}\nu_i\cdot(U'^{-1}w_i)_j>0,
        \]
        which is a contradiction again. Therefore, $b\in\cone U$ and $U\in\mathcal U$ by \cref{claim:generating-bases}.
    \end{claimproof}

    We are now ready to prove \cref{claim:induction} with induction. Here, we will follow the proof of Theorem 4.1 of~\cite{DBLP:conf/esa/CslovjecsekEPVW21}.

    \begin{claimin}
        $\mathcal D_{r,F}(b)$ is a faithful decomposition of $b$ w.r.t.\ $D$.
        \label{claim:induction}
    \end{claimin}

    \begin{claimproof}
        We prove this by showing that the following induction hypothesis holds for all $s\in\Z_{\ge0}^{[\ell]\setminus S}$: let $p=\sum_{i\in S}\alpha_i\cdot w_i+\sum_{i\in[\ell]\setminus S}(\gamma_i+s_i)\cdot w_i$. Then $\mathcal D_{r,F}(p)$ is a faithful decomposition of $p$ w.r.t.\ $D$.
    
        Note that when $s=0$, we have that $\mathcal D_{r,F}(p)=e_q=e_p$ and the claim holds. We proceed with the induction step.
        
        Now for the induction step, let $s$ and the corresponding $p$ be given. Write $\nu_i=\alpha_i$ for $i\in S$ and $\nu_i=\gamma_i+s_i$ for $i\in[\ell]\setminus S$. Note that in both cases, $\nu_i$ is strictly positive. Let $x\in\Z_{\ge0}^t$ be a solution to $Dx=p$. We show that this solution decomposes into solutions to systems with right-hand sides as described by $\mathcal D_{r,F}(p)$. Consider the polyhedron $\{Dx'=p,x'\ge0\}$. By the Minkowski-Weyl theorem, we may write
        \[
            x=\sum_j\zeta_j\cdot u_j+\sum_j\eta_j\cdot l_j
        \]
        for vertices $u_j$ and rays $l_j$ of the polyhedron, convex multipliers $\zeta_j\ge0,\sum_j\zeta_j=1$ and conic multipliers $\eta_j\ge0$. By Caratheodory's theorem, we can assume that there are at most $t+1$ nonzero multipliers $\zeta_j$. Averaging the convex multipliers shows that there must exist a $j$ with $\zeta_j\ge1/(t+1)$. As all summands in the decomposition of $x$ are nonnegative, it follows that $x\ge1/(t+1)\cdot u_j$. Let $U$ be a basis associated with $u_j$ and let $\pi$ be the projection on the corresponding basic variables, i.e., $\pi(u_j)=U^{-1}p\ge0$. Note that $p\in\cone U$ implies that the generating basis $U$ is included in $\mathcal U$ as a result of \cref{claim:stay-in-the-cone}. Now write
        \[
            p=\sum_{i\in[\ell]}\nu_i\cdot w_i=\sum_{i\in[\ell]}\nu_i\cdot U\tilde y_i
        \]
        for integral conic multipliers $\tilde y_i\in\Z_{\ge0}^d$ such that $w_i=U\tilde y_i$, which is possible as $w_i\in\cone U$ by \cref{claim:big-generators} and \cref{lemma:intersection-integer-cone} guarantees that $w_i=\Psi\cdot v_i\in\intcone U$. It follows that
        \[
            \pi(x)\ge1/(t+1)\cdot\pi(u_j)=1/(t+1)\cdot U^{-1}p=\sum_{i\in[\ell]}\tfrac{\nu_i}{t+1}\cdot\tilde y_i.
        \]
        Let $i^*\in[\ell]\setminus S$ be so that $s_{i^*}\ge1$. Then $\nu_{i^*}=s_{i^*}+t+\alpha_{i^*}\ge t+1$ and, therefore, it follows that $\pi(x)\ge\tilde y_{i^*}$. Since the columns of $U$ are present in $D$, we can extend $\tilde y_{i^*}$ with zeroes to a solution $y\in\Z_{\ge0}^t$ to $Dy=w_{i^*}$. It holds that $Dy=w_{i^*}$ and $D(x-y)=p-w_{i^*}$ with $x=y+(x-y)$ and $p=w_{i^*}+(p-w_{i^*})$ with both solutions $y$ and $x-y$ being nonnegative and integral. This corresponds to faithfully decomposing $p$ into one small vector $w_{i^*}$ and one larger vector $p'=p-w_{i^*}$. It is immediate that $p'$ can be written as the conic combination
        \[
            p'=\sum_{i\in S}\alpha_i\cdot w_i+\sum_{i\in[\ell]\setminus S}(\gamma_i+s_i')\cdot w_i'
        \]
        in $\cone V$ for multipliers that match the induction hypothesis given by $s_i'=s_i$ for all $i\in[\ell]\setminus(S\cup\{i^*\})$ except $s_{i^*}'=s_{i^*}-1\in\Z_{\ge0}$.
                
        Therefore, the induction hypothesis applies and shows that
        \[
            1\cdot e_q+\left(\sum_{i\in[\ell]\setminus S}s_i\cdot e_{w_i}\right)-1\cdot e_{w_{i^*}}
        \]
        is a faithful decomposition of $p'$. Since $x-y$ is a solution to $D(x-y)=p'$, it decomposes into solutions to the systems given by this faithful decomposition of $p'$. It is straightforward to verify that these solutions combined with the solution $y$ to $Dy=w_{i^*}=p-p'$, show that $\mathcal D_{r,F}(p)$ is indeed a valid faithful decomposition of $p$.
    \end{claimproof}

    It remains to argue about the running time of the algorithm. First, it is immediate that the hyperplanes can be constructed in time $d^{\O(1)}\cdot\O(d\Delta)^{d^4}=\O(d\Delta)^{d^4}$.
    
    Given a face $F\in\mathcal F_H$, identifying $\mathcal U$ as described in \cref{claim:generating-bases} can be done in time $d^{\O(1)}\cdot|\dbb d\Delta|=\O(\Delta)^{d^2}$.
    
    As a $\cone U=\{x\in\R^d:U^{-1}x\ge0\}$ is simplicial for $U\in\dbb d\Delta$, the intersection cone $\bigcap_{U\in\mathcal U}\cone U$ can be described with $\O(\Delta)^{d^2}$ inequalities and is pointed. To find a generating set, it suffices to naively test whether every possible intersection of $d-1$ linearly independent equalities from the inequality description is within the intersection cone and include the suitably signed spanning vector in the generating set in this case. This shows that we may obtain $P$ from \cref{lemma:intersection-integer-cone} in $d^{\O(1)}\cdot\binom{|\mathcal U|}{d-1}\cdot|\mathcal U|=d^{\O(1)}\cdot(\O(\Delta)^{d^2})^d=\O(\Delta)^{d^3}$ time and obtain $|P|\le\O(\Delta)^{d^3}$. We can then attempt to check whether $F\subseteq\cone V$ for all linearly independent sets $V$ of vectors in $P$ as suggested in \cref{claim:big-generators} to find $V$ in time $d^{\O(1)}\cdot\sum_{i=0}^d\binom{|P|}i=d^{\O(1)}\cdot d\cdot(\O(\Delta)^{d^3})^d=\O(\Delta)^{d^4}$. Extension to the basis $\overline W$ can be done in polynomial time. Finally, \cref{claim:fixed-remainder} shows that given $\overline W$, the multipliers $\alpha,\gamma$ and index set $S$ can be computed in polynomial time.
\end{proofof}

\section{Future Directions}

We presented a new slice-wise polynomial time algorithm to solve 4-block ILPs with an improved, linear, dependence of the block dimensions on the exponent of $n$. This algorithm successfully overcomes the runtime limitations of enumerating the integer vectors in an $s$-dimensional box with side length determined by the Graver complexity or proximity, and is the first to algorithm for 4-block ILPs that does not need the coefficients in the global parts to be bounded by a parameter.

The idea is partitioning the space of the global variables into $\O_{st\Delta}(n^s)$ regions, and for each such region, solving a decomposed $n$-fold ILP where the right-hand side vectors depend affinely on $x_0$. We do so by showing that faithful decompositions can be expressed as an affine function of the right-hand side vectors of the $n$-fold as long as the global variables are within a fixed region.

Despite the progress obtained in this paper, the question of whether there exists an FPT algorithm to solve 4-block ILPs (\cref{conjecture:4-block}) remains unanswered. It appears as though a different strategy must be employed to obtain such an FPT algorithm, if it were to exist: consider a 4-block ILP where $A$ is the zero matrix, $C_i=-e_{j(i)}^\top$ for some standard basis vector $e_{j(i)}$, $D_i=\begin{pmatrix}1&-1\end{pmatrix}$ and $B_i=\begin{pmatrix}a_i&0\end{pmatrix}$ for arbitrary vectors $a_i\in\{-1,0,1\}^m$. Solving such feasibility question corresponds to finding a solution to
\begin{align}
    &a_1y_1+\dots+a_ny_n=b_0,\notag\\
    &y_i\ge\max\{0,b_i+(x_0)_{j(i)}\}\ \forall i\in[n],\label{ineq:1-1}\\
    &x_0\in\Z_{\ge0}^s,y\in\Z^n\notag.
\end{align}
Here, our algorithm needs to guess whether $(x_0)_{j(i)}\le b_i$ or $(x_0)_{j(i)}<b_i$ for all $i\in[n]$, resulting in the space of the global variables being partitioned into an $s$-dimensional axis-aligned grid that may have up to $\Omega_s(n^s)$ many cells. We note that if the inequality in (\ref{ineq:1-1}) is replaced with equality, finding a solution becomes $W[1]$-hard. Roughly sketching the idea, the $s$-clique problem can be encoded letting each global variable $(x_0)_j$ pick a vertex $v\in\{1,2,\dots,|V|\}$ to be in the clique and adding many duplicate bricks to mimic an ILP with $m=\O(s^2)$ constraints and unary encoded coefficients on binary variables $z_{jv}$ indicating whether $(x_0)_j=v$, which can encode the required clique constraints as shown in~\cite{DBLP:journals/ai/DvorakEGKO21}. Note that, in standard form, (\ref{ineq:1-1}) corresponds to $y_i-z_i=b_i+(x_0)_{j(i)}$ for a slack variable $z_i\in\Z_{\ge0}$, which shows that a slack of $s$ in (\ref{ineq:1-1}) corresponds to adding the Graver basis element $g=\begin{pmatrix}1&1\end{pmatrix}$ a total of $p_g^i=s$ times in the decomposed \cref{ilp:decomposed-4-block}. Therefore, we believe that a natural next question in order to resolve \cref{conjecture:4-block} is to study if and how the convexity of the inequality in (\ref{ineq:1-1}) or the freedom of adding Graver basis elements in \cref{ilp:decomposed-4-block} can be exploited.

Finally, we note that the Graver complexity and proximity bounds from~\cite{DBLP:journals/disopt/ChenCZ22,DBLP:journals/mp/HemmeckeKW14,DBLP:journals/mp/OertelPW24} can be used to minimize separable convex objective functions as described in~\cite{DBLP:journals/mp/HemmeckeKW14}, whereas it is unclear whether the algorithm underlying \cref{theorem:uniform-4-block-xp} can similarly be generalized. Therefore, it remains an interesting open question whether one can optimize such nonlinear objective functions over 4-block integer programs in time $\O_{dmst\soverline\Delta}(n^{s+\O(1)})$.

\bibliography{bib}

@article{DBLP:journals/ai/DvorakEGKO21,
  author = {Pavel Dvořák and Eduard Eiben and Robert Ganian and Dusan Knop and Sebastian Ordyniak},
  title = {The complexity landscape of decompositional parameters for {ILP:} Programs with few global variables and constraints},
  journal = {Artificial Intelligence},
  volume = {300},
  pages = {103561},
  year = {2021},
  doi = {10.1016/j.artint.2021.103561}
}

@article{DBLP:journals/mp/HemmeckeKW14,
  author = {Raymond Hemmecke and Matthias K{\"{o}}ppe and Robert Weismantel},
  title = {Graver basis and proximity techniques for block-structured separable convex integer minimization problems},
  journal = {Mathematical Programming},
  volume = {145},
  number = {1-2},
  pages = {1--18},
  year = {2014},
  doi = {10.1007/s10107-013-0638-z}
}

@inproceedings{DBLP:conf/icalp/EisenbrandHK18,
  author = {Friedrich Eisenbrand and Christoph Hunkenschr{\"{o}}der and Kim{-}Manuel Klein},
  title = {Faster Algorithms for Integer Programs with Block Structure},
  booktitle = {Proceedings of the 45th International Colloquium on Automata, Languages, and Programming, {ICALP}},
  pages = {49:1--49:13},
  year = {2018},
  doi = {10.4230/LIPIcs.ICALP.2018.49}
}

@article{DBLP:journals/mp/CookGST86,
  author = {William J. Cook and Albertus M. H. Gerards and Alexander Schrijver and {\'{E}}va Tardos},
  title = {Sensitivity theorems in integer linear programming},
  journal = {Mathematical Programming},
  volume = {34},
  number = {3},
  pages = {251--264},
  year = {1986},
  doi = {10.1007/BF01582230}
}

@article{DBLP:journals/theoretics/CslovjecsekKLPP25,
  author = {Jana Cslovjecsek and Martin Kouteck{\'{y}} and Alexandra Lassota and Michal Pilipczuk and Adam Polak},
  title = {Parameterized algorithms for block-structured integer programs with large entries},
  journal = {TheoretiCS},
  volume = {4},
  year = {2025},
  doi = {10.46298/THEORETICS.25.15},
  note = {Previously appeared in SODA 2024}
}

@article{DBLP:journals/ior/Tardos86,
  author = {{\'{E}}va Tardos},
  title = {A Strongly Polynomial Algorithm to Solve Combinatorial Linear Programs},
  journal = {Opererations Research},
  volume = {34},
  number = {2},
  pages = {250--256},
  year = {1986},
  doi = {10.1287/OPRE.34.2.250}
}

@misc{eisenbrand2025parameterizedlinearformulationinteger,
  title = {A parameterized linear formulation of the integer hull},
  author = {Friedrich Eisenbrand and Thomas Rothvoss},
  year = {2025},
  eprint = {2501.02347},
  archiveprefix = {arXiv},
  primaryclass = {cs.CC}
}

@inproceedings{DBLP:conf/ipco/HemmeckeKW10,
  author = {Raymond Hemmecke and Matthias K{\"{o}}ppe and Robert Weismantel},
  title = {A Polynomial-Time Algorithm for Optimizing over \emph{N}-Fold 4-Block Decomposable Integer Programs},
  booktitle = {Proceedings of the 14th International Conference on Integer Programming and Combinatorial Optimization, {IPCO}},
  pages = {219--229},
  year = {2010},
  doi = {10.1007/978-3-642-13036-6\_17}
}

@article{DBLP:journals/mp/JansenKMR22,
  author = {Klaus Jansen and Kim{-}Manuel Klein and Marten Maack and Malin Rau},
  title = {Empowering the configuration-IP: new {PTAS} results for scheduling with setup times},
  journal = {Mathematical Programming},
  volume = {195},
  number = {1},
  pages = {367--401},
  year = {2022},
  doi = {10.1007/S10107-021-01694-3}
}

@inproceedings{DBLP:conf/isaac/FellowsLMRS08,
  author = {Michael R. Fellows and Daniel Lokshtanov and Neeldhara Misra and Frances A. Rosamond and Saket Saurabh},
  title = {Graph Layout Problems Parameterized by Vertex Cover},
  booktitle = {Proceedings of the 19th International Symposium on Algorithms and Computation, {ISAAC}},
  pages = {294--305},
  year = {2008},
  doi = {10.1007/978-3-540-92182-0_28}
}

@article{DBLP:journals/dam/FialaGKKK18,
  author = {Jir{\'{\i}} Fiala and Tomas Gavenciak and Dusan Knop and Martin Kouteck{\'{y}} and Jan Kratochv{\'{\i}}l},
  title = {Parameterized complexity of distance labeling and uniform channel assignment problems},
  journal = {Discrete Applied Mathematics},
  volume = {248},
  pages = {46--55},
  year = {2018},
  doi = {10.1016/j.dam.2017.02.010}
}

@article{DBLP:journals/teco/KnopKM20,
  author = {Dusan Knop and Martin Kouteck{\'{y}} and Matthias Mnich},
  title = {Voting and Bribing in Single-Exponential Time},
  journal = {{ACM} Transactions on Economics and Computation},
  volume = {8},
  number = {3},
  pages = {12:1--12:28},
  year = {2020},
  doi = {10.1145/3396855}
}

@article{bartholdi1989voting,
  title = {Voting schemes for which it can be difficult to tell who won the election},
  author = {Bartholdi, John and Tovey, Craig A and Trick, Michael A},
  journal = {Social Choice and Welfare},
  volume = {6},
  number = {2},
  pages = {157--165},
  year = {1989},
  doi = {10.1007/BF00303169}
}

@article{ermolieva2023connections,
  title = {Connections between Robust Statistical Estimation, Robust Decision-Making with Two-Stage Stochastic Optimization, and Robust Machine Learning Problems},
  author = {Ermolieva, T. and Ermoliev, Y and Havlik, P. and Lessa-Derci-Augustynczik, A. and Komendantova, N. and Kahil, T. and Balkovic, J. and Skalsky, R. and Folberth, C. and Knopov, P. S. and Wang, G.},
  journal = {Cybernetics and Systems Analysis},
  volume = {59},
  number = {3},
  pages = {385--397},
  year = {2023},
  doi = {10.1007/s10559-023-00573-3}
}

@article{DBLP:journals/disopt/GavenciakKK22,
  author = {Tom\'a\v{s} Gaven\v{c}iak and Martin Kouteck{\'{y}} and Du\v{s}an Knop},
  title = {Integer programming in parameterized complexity: Five miniatures},
  journal = {Discrete Optimization},
  volume = {44},
  number = {Part 1},
  pages = {100596},
  year = {2022},
  doi = {10.1016/j.disopt.2020.100596}
}

@inproceedings{DBLP:conf/esa/CslovjecsekEPVW21,
  author = {Jana Cslovjecsek and Friedrich Eisenbrand and Michal Pilipczuk and Moritz Venzin and Robert Weismantel},
  title = {Efficient Sequential and Parallel Algorithms for Multistage Stochastic Integer Programming Using Proximity},
  booktitle = {Proceedings of the 29th Annual European Symposium on Algorithms, {ESA}},
  volume = {204},
  pages = {33:1--33:14},
  year = {2021},
  doi = {10.4230/LIPICS.ESA.2021.33}
}

@article{DBLP:journals/siamcomp/EdelsbrunnerOS86,
  author = {Herbert Edelsbrunner and Joseph O'Rourke and Raimund Seidel},
  title = {Constructing Arrangements of Lines and Hyperplanes with Applications},
  journal = {{SIAM} Journal on Computing},
  volume = {15},
  number = {2},
  pages = {341--363},
  year = {1986},
  doi = {10.1137/0215024}
}

@article{DBLP:journals/combinatorica/FrankT87,
  author = {Andr{\'{a}}s Frank and {\'{E}}va Tardos},
  title = {An application of simultaneous Diophantine approximation in combinatorial optimization},
  journal = {Combinatorica},
  volume = {7},
  number = {1},
  pages = {49--65},
  year = {1987},
  doi = {10.1007/BF02579200}
}

@article{DBLP:journals/mor/Kannan87,
  author = {Ravi Kannan},
  title = {Minkowski's Convex Body Theorem and Integer Programming},
  journal = {Mathematics of Operations Research},
  volume = {12},
  number = {3},
  pages = {415--440},
  year = {1987},
  doi = {10.1287/MOOR.12.3.415},
}

@article{DBLP:journals/mp/OertelPW24,
  author = {Timm Oertel and Joseph Paat and Robert Weismantel},
  title = {A colorful Steinitz Lemma with application to block-structured integer programs},
  journal = {Mathematical Programming},
  volume = {204},
  number = {1},
  pages = {677--702},
  year = {2024},
  doi = {10.1007/S10107-023-01971-3}
}

@article{Hanson_1972,
  title = {On the Product of the Primes},
  volume = {15},
  number = {1},
  journal = {Canadian Mathematical Bulletin},
  author = {Hanson, Denis},
  year = {1972},
  pages = {33–-37},
  doi = {10.4153/CMB-1972-007-7}
}

@inproceedings{DBLP:conf/soda/CslovjecsekEHRW21,
  author = {Jana Cslovjecsek and Friedrich Eisenbrand and Christoph Hunkenschr{\"{o}}der and Lars Rohwedder and Robert Weismantel},
  title = {Block-Structured Integer and Linear Programming in Strongly Polynomial and Near Linear Time},
  booktitle = {Proceedings of the 2021 {ACM-SIAM} Symposium on Discrete Algorithms, {SODA}},
  pages = {1666--1681},
  year = {2021},
  doi = {10.1137/1.9781611976465.101}
}

@inproceedings{DBLP:conf/esa/0009K0S20,
  author = {Lin Chen and Martin Kouteck{\'{y}} and Lei Xu and Weidong Shi},
  title = {New Bounds on Augmenting Steps of Block-Structured Integer Programs},
  booktitle    = {Proceedings of the 28th Annual European Symposium on Algorithms, {ESA}},
  pages = {33:1--33:19},
  year = {2020},
  doi = {10.4230/LIPICS.ESA.2020.33}
}

@article{DBLP:journals/mp/Klein22,
  author = {Kim{-}Manuel Klein},
  title = {About the complexity of two-stage stochastic IPs},
  journal = {Mathematical Programming},
  volume = {192},
  number = {1},
  pages = {319--337},
  year = {2022},
  doi = {10.1007/S10107-021-01698-Z}
}

@article{DBLP:journals/disopt/ChenCZ22,
  author = {Hua Chen and Lin Chen and Guochuan Zhang},
  title = {Block-structured integer programming: Can we parameterize without the largest coefficient?},
  journal = {Discrete Optimization},
  volume = {46},
  pages = {100743},
  year = {2022},
  doi = {10.1016/J.DISOPT.2022.100743}
}

@article{DBLP:journals/scheduling/KnopK18,
  author = {Dusan Knop and Martin Kouteck{\'{y}}},
  title = {Scheduling meets n-fold integer programming},
  journal = {Journal of Scheduling},
  volume = {21},
  number = {5},
  pages = {493--503},
  year = {2018},
  doi = {10.1007/S10951-017-0550-0}
}

@article{DBLP:journals/mp/ChenCZ24,
  author = {Hua Chen and Lin Chen and Guochuan Zhang},
  title = {{FPT} algorithms for a special block-structured integer program with applications in scheduling},
  journal = {Mathematical Programming},
  volume  = {208},
  number = {1},
  pages = {463--496},
  year = {2024},
  doi = {10.1007/S10107-023-02046-Z}
}

@inproceedings{pottier1991minimal,
  author = {Lo{\"i}c Pottier},
  title = {Minimal solutions of linear diophantine systems : bounds and algorithms},
  booktitle = {Proceedings of the International Conference on Rewriting Techniques and Applications, {RTA}},
  year = {1991},
  pages = {162--173},
  doi = {10.1007/3-540-53904-2_94}
}

@book{DBLP:series/eatcs/Edelsbrunner87,
  author = {Herbert Edelsbrunner},
  title = {Algorithms in Combinatorial Geometry},
  series = {{EATCS} Monographs on Theoretical Computer Science},
  volume = {10},
  publisher = {Springer},
  year = {1987},
  doi = {10.1007/978-3-642-61568-9}
}

@article{DBLP:journals/siamcomp/EdelsbrunnerSS93,
  author = {Herbert Edelsbrunner and Raimund Seidel and Micha Sharir},
  title = {On the Zone Theorem for Hyperplane Arrangements},
  journal = {{SIAM} Journal on Computing},
  volume = {22},
  number = {2},
  pages = {418--429},
  year = {1993},
  doi = {10.1137/0222031}
}

@book{grunbaum1967convex,
  author = {Branko Gr{\"u}nbaum},
  edition = {Second},
  title = {Convex polytopes},
  series = {Graduate Texts in Mathematics},
  volume = {221},
  year = {2003},
  publisher = {Springer},
  doi = {10.1007/978-1-4613-0019-9}
}

@inproceedings{DBLP:conf/soda/BachERW25,
  author = {Eleonore Bach and Friedrich Eisenbrand and Thomas Rothvoss and Robert Weismantel},
  title = {Forall-exist statements in pseudopolynomial time},
  booktitle = {Proceedings of the 2025 {ACM-SIAM} Symposium on Discrete Algorithms, {SODA}},
  pages = {2225--2233},
  year = {2025},
  doi = {10.1137/1.9781611978322.73}
}

@inproceedings{DBLP:conf/icalp/KouteckyLO18,
  author = {Martin Kouteck{\'{y}} and Asaf Levin and Shmuel Onn},
  title = {A Parameterized Strongly Polynomial Algorithm for Block Structured Integer Programs},
  booktitle = {Proceedings of the 45th International Colloquium on Automata, Languages, and Programming, {ICALP}},
  pages = {85:1--85:14},
  year = {2018},
  doi = {10.4230/LIPICS.ICALP.2018.85}
}

\end{document}